\documentclass[lettersize,journal]{IEEEtran}
\usepackage{amsmath,amsfonts}
\usepackage{algorithmic}
\usepackage{algorithm}
\usepackage{array}
\usepackage[caption=true,font=footnotesize,labelfont=rm,textfont=rm]{subfig}
\usepackage{textcomp}
\usepackage{stfloats}
\usepackage{url}
\usepackage{verbatim}
\usepackage{graphicx}
\usepackage{cite}
\usepackage{xcolor}
\usepackage[acronym]{glossaries}
\usepackage{siunitx}
\usepackage{graphicx}
\usepackage{amssymb}  
\usepackage{caption}
\usepackage{flushend}
\usepackage{booktabs} 
\usepackage{mathtools}

\usepackage{pgfplots}
  \pgfplotsset{compat=newest}
  \usetikzlibrary{plotmarks}
  \usetikzlibrary{arrows.meta}
  \usepgfplotslibrary{patchplots}
  \usepackage{grffile}

  \pgfplotsset{plot coordinates/math parser=false}
  \newlength\figureheight
  \newlength\figurewidth

\makeglossaries
\newacronym{pms}{PMS}{power management system}
\newacronym{dgu}{DGU}{distributed generation unit}
\newacronym{epu}{EPU}{electric propulsion unit}
\newacronym{ems}{EMS}{energy management system}
\newacronym{mg}{MG}{microgrid}
\newacronym{epds}{EPDS}{electrical power distribution system}
\newacronym{pcc}{PCC}{point of common coupling}
\newacronym{mea}{MEA}{more-electric aircraft}
\newacronym{cpl}{CPL}{constant power load}
\newacronym{tgs}{TGS}{turbo generator system}
\newacronym{ess}{ESS}{energy storage system}
\newacronym{pmsm}{PMSM}{permanent magnet synchronous motor}
\newacronym{phil}{PHIL}{power hardware-in-the-loop}
\newacronym{uav}{UAV}{unmanned aerial vehicle}
\newacronym{flc}{FLC}{fuzzy logic controller}
\newacronym{soc}{SOC}{state-of-charge}
\newacronym{rtc}{RTC}{real-time computer}

\usepackage{tikz}
\usepackage{tikz}

\usetikzlibrary{shapes.geometric}
\usetikzlibrary{shapes.misc}
\usetikzlibrary{shapes.symbols}
\usetikzlibrary{arrows, arrows.meta, positioning, 3d,calc}
\usetikzlibrary{decorations.pathreplacing}
\tikzset{
  >=stealth',
}

\usepackage[straight voltages, american]{circuitikz}
\ctikzset{resistors/scale=0.7,
  capacitors/scale=0.7,
  inductors/scale=0.7,
  instruments/scale=0.5,
}

\makeatletter
\tikzset{self loop/.style =  {to path={
  \pgfextra{}
  [looseness=12,min distance=10mm]
  \tikz@to@curve@path},font=\sffamily\small
  }}
\makeatother

\makeatletter
\tikzset{%
  block/.style      = {draw, thick, rectangle, minimum height = .5cm, minimum width = 1cm, align=center, font=\footnotesize},
  hub/.style        = {draw, circle, fill=black, minimum size=.10cm, inner sep=0.00cm},
  symhub/.style     = {hub, fill=none, draw, inner sep=0.005cm, scale=.5},
  connection/.style = {draw, ->},
  conlabel/.style   = {scale=0.75},
  ede/.style = {circle, fill=lightgray, font=\scriptsize, inner sep=0.05cm, minimum size=.25cm},
  eae/.style = {draw, circle, font=\scriptsize, inner sep=0.05cm, minimum size=.25cm},
}
\makeatother

\makeatletter
\newcommand{\pgfwest}{%
    \pgf@process{\northeast}%
    \pgf@ya=.5\pgf@y%
    \pgf@process{\southwest}%
    \pgf@y=.5\pgf@y%
    \advance\pgf@y by \pgf@ya%
}
\newcommand{\pgfeast}{%
    \pgf@process{\southwest}%
    \pgf@ya=.5\pgf@y%
    \pgf@process{\northeast}%
    \pgf@y=.5\pgf@y%
    \advance\pgf@y by \pgf@ya%
}
\newcommand{\pgfsouth}{%
    \pgf@process{\northeast}%
    \pgf@xa=.5\pgf@x%
    \pgf@process{\southwest}%
    \pgf@x=.5\pgf@x%
    \advance\pgf@x by \pgf@xa%
}
\newcommand{\pgfnorth}{%
    \pgf@process{\southwest}%
    \pgf@xa=.5\pgf@x%
    \pgf@process{\northeast}%
    \pgf@x=.5\pgf@x%
    \advance\pgf@x by \pgf@xa%
}

\long\def\pgfshapeaddanchor#1#2{%
{%
  \def\pgf@sm@shape@name{#1}%
  \let\anchor=\pgf@sh@anchor%
  #2}%
}
\pgfshapeaddanchor{rectangle}{%
  \anchor{west north west}{%
    \pgf@process{\northeast}%
    \pgf@ya=1.5\pgf@y%
    \pgf@process{\southwest}%
    \pgf@y.5\pgf@y%
    \advance\pgf@y by \pgf@ya%
    \pgf@y=.5\pgf@y%
  }%
  \anchor{north north west}{%
    \pgf@process{\southwest}%
    \pgf@xa=1.5\pgf@x%
    \pgf@process{\northeast}%
    \pgf@x.5\pgf@x%
    \advance\pgf@x by \pgf@xa%
    \pgf@x=.5\pgf@x%
  }%
  \anchor{west south  west}{%
    \pgf@process{\northeast}%
    \pgf@ya=.5\pgf@y%
    \pgf@process{\southwest}%
    \pgf@y=1.5\pgf@y%
    \advance\pgf@y by \pgf@ya%
    \pgf@y=.5\pgf@y%
  }%
  \anchor{south south  west}{%
    \pgf@process{\northeast}%
    \pgf@xa=.5\pgf@x%
    \pgf@process{\southwest}%
    \pgf@x=1.5\pgf@x%
    \advance\pgf@x by \pgf@xa%
    \pgf@x=.5\pgf@x%
  }%
  \anchor{west -1}{%
    \pgfwest%
    \advance\pgf@y by -0.5 cm%
  }%
  \anchor{west 1}{%
    \pgfwest%
    \advance\pgf@y by 0.5 cm%
  }%
  \anchor{west A}{%
    \pgfwest%
    \advance\pgf@y by 0.25 cm%
  }
  \anchor{west B}{%
    \pgfwest%
    \advance\pgf@y by -0.25 cm%
  }%
  \anchor{east -1}{%
    \pgfeast%
    \advance\pgf@y by -0.5 cm%
  }%
  \anchor{east 1}{%
    \pgfeast%
    \advance\pgf@y by 0.5 cm%
  }%
  \anchor{east A}{%
    \pgfeast%
    \advance\pgf@y by 0.25 cm%
  }
  \anchor{east B}{%
    \pgfeast%
    \advance\pgf@y by -0.25 cm%
  }%
  \anchor{north -1}{%
    \pgfnorth%
    \advance\pgf@x by -0.5 cm%
  }%
  \anchor{north 1}{%
    \pgfnorth%
    \advance\pgf@x by 0.5 cm%
  }%
  \anchor{north A}{%
    \pgfnorth%
    \advance\pgf@x by -0.25 cm%
  }
  \anchor{north B}{%
    \pgfnorth%
    \advance\pgf@x by 0.25 cm%
  }%
  \anchor{south -1}{%
    \pgfsouth%
    \advance\pgf@x by -0.5 cm%
  }%
  \anchor{south 1}{%
    \pgfsouth%
    \advance\pgf@x by 0.5 cm%
  }%
  \anchor{south A}{%
    \pgfsouth%
    \advance\pgf@x by 0.25 cm%
  }
  \anchor{south B}{%
    \pgfsouth%
    \advance\pgf@x by -0.25 cm%
  }%
}
\makeatother

\newtheorem{proposition}{Proposition}

\newtheorem{remark}{Remark}

\newtheorem{assumption}{Assumption}
\newtheorem{objective}{Objective}

\newtheorem{proof}{Proof}

\newcommand{\coo}{\ensuremath{\mathrm{CO_2}}}

\begin{document}

\title{Distributed Adaptive Control for DC Power Distribution in Hybrid-Electric Aircraft: Design and Experimental Validation$^*$}

\author{Wasif H. Syed$^{1}$, Juan E. Machado$^{1}$, Hans Würfel$^{2}$, Ekrem Hanli$^{3}$ and Johannes Schiffer$^{1,4}$
\thanks{$^*$ An abridged version of this article was presented at the 2024 American Control Conference (ACC), Toronto, CA, July 2024. (See \cite{syed}.)  This work is partially supported by the Bundesministerium für Wirtschaft und Klimaschutz (BMWK), Germany, under LuFo VI - 2 joint project ``Safe and reliable electrical and thermal networks for hybrid-electric drive systems (ETHAN, project number: 20L2103F1". This research also received funding from the German Federal Government, the Federal Ministry of Education and Research, and the State of Brandenburg within the framework of the joint project EIZ: Energy Innovation Center (project numbers 85056897 and 03SF0693A).}
\thanks{$^{1}$ Chair of Control Systems and Network Control Technology, Brandenburg University of Technology Cottbus-Senftenberg, 03046, Germany, 
        {\tt\small syed@b-tu.de (Wasif H. Syed)}, {\tt\small machadom@b-tu.de (Juan E. Machado)}, {\tt\small schiffer@b-tu.de (Johannes Schiffer).}} %
\thanks{$^{2}$  Potsdam Institute for Climate Impact Research, Germany, 
        {\tt\small wuerfel@pik-potsdam.de (Hans Würfel).}}%
\thanks{$^{3}$ Rolls-Royce Deutschland Ltd \& Co,  \\
        {\tt\small ekrem.hanli@rolls-royce.com (Ekrem Hanli).}}%
\thanks{$^{4}$ Fraunhofer Research Institution for Energy Infrastructures and Geotechnologies, Cottbus, 03046, Germany.  
        }%
}

\markboth{~}%
{Shell \MakeLowercase{\textit{et al.}}: A Sample Article Using IEEEtran.cls for IEEE Journals}

\maketitle  

\begin{abstract}
To reduce \textbf{\coo } emissions and tackle increasing fuel costs, the aviation industry is swiftly moving towards the electrification of aircraft. From the viewpoint of systems and control, a key challenge brought by this transition corresponds to the management and safe operation of the propulsion system's onboard electrical power distribution network. In this work, for a series-hybrid-electric propulsion system, we propose a distributed adaptive controller for regulating the voltage of a DC  bus that energizes the electricity-based propulsion system. The proposed controller---whose design is based on principles of back-stepping, adaptive, and passivity-based control techniques---also enables the proportional sharing of the electric load among multiple converter-interfaced sources, which reduces the likelihood of over-stressing individual sources. Compared to existing control strategies, our method ensures stable,  convergent, and accurate voltage regulation and load-sharing even if the effects of power lines of {\em unknown} resistances and inductances are considered. The
performance of the proposed control scheme is experimentally
validated and compared to state-of-the-art controllers in a power
hardware-in-the-loop (PHIL) environment.

\end{abstract}


\section{Introduction}
\subsection{Motivation}

Currently, the propulsion systems of commercial aircraft are predominantly powered by fossil fuels \cite{WorldData}. Opportunely, the transition towards a more sustainable aviation sector is underway. On the one hand, this is motivated by the urgency of reducing \coo{} emissions to the atmosphere,  as commercial aircraft operations are currently liable for 2.5\% of global emissions \cite{WorldData}. On the other hand, the costs of aviation fuels are expected to increase due to regulatory measures \cite{FuelCost}. 
The transition towards a more sustainable aviation sector entails further electrification of aircraft and their propulsion systems, as it would substantially increase the capabilities to utilize renewable sources either indirectly, i.e.,  by storing their energy in onboard batteries,  or directly, e.g.,  via the co-operation of both conventional and varied sustainable sources, such as fuel cells in hybrid propulsion configurations \cite{Schefer}.

In this work, the focus is on series-hybrid-electric propulsion systems, which are considered to be well-suited for larger aircraft \cite{Schefer,XieYe} and for multi-rotor aircraft \cite{XieYe}. Unlike conventional setups, this type of aircraft propulsion system is based on both electrical and mechanical energy. In particular, the system realizing electric propulsion consists of multiple converter-interfaced sources, which are connected in parallel to a common DC bus from which the \gls{epu} (i.e., an inverter-interfaced electric motor) is then energized \cite{XieYe}. The correct operation of the propulsion load strongly depends on the DC bus having a voltage near a prespecified nominal value \cite{ZhangSau}. Critically, load changes or the connection or disconnection of electric sources to the DC bus may lead to undesired voltage variations. Moreover, to prevent the over-stressing of any given source beyond its capacity, the total electric load is desired to be proportionally shared among the participating sources \cite{Noroozi2}. {\color{black} Consequently, the central research problem addressed in this paper is the design of a reliable and robust control system (i.e., robust to system parameter uncertainties) that simultaneously guarantees voltage regulation at the DC bus and proportional load sharing among sources.}

{\color{black}

\subsection{Literature Review}

The \gls{epds} of hybrid-electric propulsion systems comprise an islanded power system with generation units, power converters, and electric loads; thus, it can be modeled as an onboard DC \gls{mg} \cite{Braitor}. In this context, numerous solutions have been proposed in the DC microgrid literature to address voltage regulation tasks while achieving proportional load sharing among the generation units; see, e.g., the references reviewed in \cite{MengL2}. In particular, in \cite{Trip2},  proportional current sharing and average voltage regulation (i.e., voltage balancing) are achieved based on the presumption that the parasitic resistances associated with power converters at the generation side are either insignificant or accurately known, such that their effect is mitigated through feedforward compensation.  Nonetheless, in aeronautic applications, the use of relatively long cables introduces non-negligible inductive and resistive effects, which have to be considered in the stability analysis~\cite{Magne}. 
Voltage balancing and proportional current sharing are addressed in \cite{Nahata} for the case of non-negligible parasitic resistances via a passivity-based control scheme. Notably,  the controller's implementation requires knowledge of a lower bound on the resistances but not necessarily their exact value.  
 In \cite{HandongBai}, a distributed control scheme for a multi-bus DC MG is proposed, which simultaneously addresses approximate current sharing and voltage regulation through a trade-off factor. Additionally, sufficient conditions for closed-loop stability are provided. However, in the present work, the focus is to regulate the voltage of a common DC bus that energizes the \gls{epu} of a series-hybrid-electric propulsion system. As a result, there is no need for a trade-off between current sharing accuracy and voltage regulation since these objectives are not in conflict in this case. 
To achieve the objective of current sharing in DC \gls{mg}s, a centralized model predictive control-based consensus algorithm is proposed in \cite{Noroozi2}. The controller also guarantees that the voltage of the DC buses stays within the predefined constraints.  However, the scheme requires accurate knowledge of system parameters for proper operation. A decentralized control scheme for DC \gls{mg}s is proposed in \cite{Tucci} to achieve voltage stabilization with plug-and-play capability for DGUs. Importantly, the controller design does not require knowledge of the power-line parameters. Closed-loop stability is established using a Lyapunov function and LaSalle’s invariance principle. However, the approach does not address the problem of proportional load sharing among DGUs.

While the DC microgrid literature provides valuable insights into voltage regulation and proportional load sharing in networked DC systems, analogous control challenges have also been investigated in the aircraft \gls{epds} literature. Accordingly, the control of hybrid or all-electric propulsion system EPDS has been explored in \cite{Karunarathne, Hoenicke, bastos, Ahmed}.
In \cite{Karunarathne}, a power management system is proposed for an unmanned aerial vehicle's hybrid propulsion system driven by a fuel cell and battery. It considers the goal of managing dynamic power requirements during various flight phases, which is addressed by employing a fuzzy logic controller that regulates the power flow between the fuel cell and the battery. The fuzzy logic controller monitors the system's load demand and the battery's state-of-charge to adjust the power distribution between the sources. While this approach successfully adapts the power flow to varying operational demands, it does not explicitly address the voltage regulation problem. 
In \cite{Hoenicke}, a control strategy for hybrid-electric aircraft is presented, which manages a system composed of a fuel cell and a battery that energizes a DC bus from which propulsion is fed. This approach ensures voltage regulation and safe battery charging in various operation modes.
In \cite{bastos}, also in the context of a hybrid-electric propulsion system, a PI controller for a DC–DC converter is proposed to regulate the power flow between an \gls{ess} and a common DC-link to which the power sources (i.e., the \gls{ess} and a generator) as well as the \gls{epu} are connected. In this case, the proposed scheme ensures stable voltage regulation of the common DC bus.
In \cite{Ahmed}, a power management solution, also based on a PI control, is presented. The scheme achieves power flow management between a fuel cell, a battery, and the propulsion unit while maintaining precise voltage regulation across a potential flight profile. 
While the works  \cite{Hoenicke, bastos, Ahmed} address the problem of voltage regulation within the \gls{epds} of hybrid-electric propulsion systems, they are limited to coordinating load-sharing between only two energy sources. Moreover, these works do not explicitly seek to establish consensus on the current injections of the energy sources, i.e., proportional load-sharing. The lack of such coordination is undesirable, as it may lead to over-stressing one or more energy sources. The work in \cite{DoffSotta} proposes an MPC-based energy management strategy for hybrid-electric aircraft to optimally allocate power between propulsion components. However, it does not address the control of the \gls{epds}, including DC-bus voltage regulation, proportional load sharing among sources, or stability of the underlying electrical network.

}

\subsection{Contributions}

{\color{black}

In this work, inspired by \cite{Machado2,Trip2,Nahata}, we propose a distributed controller for voltage regulation and proportional load-sharing in a DC power distribution network representing the propulsion system of a series-hybrid-electric aircraft. The controller combines elements of back-stepping, adaptive, and passivity-based control, and guarantees local convergence of the closed-loop trajectories to an equilibrium satisfying the control objectives. Compared to \cite{Machado2,Trip2,Nahata}, the network topology considered here is structurally different, as multiple electric sources are connected to a single non-actuated load bus representing the main DC link of the propulsion system. Consequently, existing convergence analyses for multi-PCC microgrids are not directly applicable. This work, therefore, develops a control framework and corresponding convergence analysis tailored to this architecture. The main contributions are summarized as follows:

\begin{itemize}

\item \textbf{Control design for a single-load-bus architecture.} 
We develop a distributed adaptive control framework for DC power distribution networks where multiple DGUs are connected to a single non-actuated load bus. This topology is representative of the \gls{epds} of series-hybrid-electric aircraft.

\item \textbf{Reduced parameter knowledge with adaptive resistance estimation.} 
Unlike our preliminary work~\cite{syed} and \cite{Machado2}, which require exact knowledge of the line inductances, the proposed controller does not require prior knowledge of either the power-line resistances or inductances. Instead, it incorporates an adaptive mechanism that estimates the line resistances online and compensates for uncertain inductive dynamics. 

\item \textbf{Stability and convergence analysis for the considered topology.} 
We develop a convergence analysis tailored to the considered system architecture. Compared to our preliminary work~\cite{syed}, the analysis is extended to accommodate the reduced parameter knowledge. 

\item \textbf{Experimental validation in a \gls{phil} environment.} 
The proposed control framework is validated through laboratory experiments using a \gls{phil} testing setup. The experiments demonstrate robust performance under practical electrical disturbances inherent to \gls{phil} environments, such as noise, delays, and parameter uncertainties. This validation study extends our preliminary work~\cite{syed}.

\end{itemize}

}

\subsection{Organization of the Paper}
The paper is structured as follows. In Section \ref{Sec:SHEPS}, we discuss the typical architecture of a series-hybrid-electric propulsion system and its electrical components. In Section \ref{Sec:modelAndPf}, we derive an analytical model of an \gls{epds} of a series-hybrid-electric propulsion system and formalize
the control objectives. In Section \ref{Sec:Control}, we present the proposed control scheme and an asymptotic stability proof of the closed-loop system's desired equilibrium. Experimental and comparative results are presented in Section~\ref{Sec:Experiments}. Finally, in Section \ref{Sec:Conclusion}, we summarize our key findings and identify potential research directions for future work.

\subsection{Notation}
The following notation is adopted in the paper: the symbol  $\mathbb{R}$ denotes the set of real numbers;  $x = \text{col} (x_i)$ is a column vector with entries $x_i$, where $x_i$ can be either a scalar or a vector; $\text{diag}(x_i)$ is a diagonal matrix with diagonal entries $x_i$; $\boldsymbol{1}_{n} \in \mathbb{R}^n$ is the vector with all entries being equal to one, $\boldsymbol{0}_{n} \in \mathbb{R}^n$ is the vector with all entries being equal to zero and $\boldsymbol{0}_{n \times m} \in \mathbb{R}^{n \times m}$ is the matrix with all entries being equal to zero; $\mathbb{I}_n$ denotes the  $n \times n$ identity matrix; the expression $x^\top M x$ is equivalently represented as $\|x\|^2_{M}$, where $M$ is any positive definite, symmetric matrix and $x$ is a column vector. Moreover, in Tables~ \ref{table:abbreviations} and \ref{table:symbols}, we summarize the main abbreviations and symbols, respectively, used throughout the paper.

\begin{table}[ht]
\centering 
\caption{List of Abbreviations}
\label{table:abbreviations}
\begin{tabular}{ll}
\toprule
Abbreviation & Definition \\
\midrule
DGU & Distributed generation unit \\
EPDS & Electrical power distribution system \\
EPU & Electric propulsion unit \\
ESS & Energy storage system \\
MG & Microgrid \\
PCC & Point of common coupling \\
PHIL & Power hardware-in-the-loop \\
PMS & Power management system \\
TGS & Turbo generator system \\
\bottomrule
\end{tabular}
\end{table}

\begin{table}[ht]
\centering 
\caption{List of Symbols}
\label{table:symbols}
\begin{tabular}{ll}
\toprule
Notation & Definition \\
\midrule
$C_{\mathrm{dc}}$ & Capacitance of PCC \\
$\mathcal{G}$ & Connected graph to represent the EPDS topology \\
$\mathcal{G}^{\mathrm{com}}$ & Connected graph representing communication network\\ 
                 & among the DGUs  \\
$\mathcal{E}$ & Set containing the power lines \\
$\mathcal{E}^{\text{com}}$ & Set containing communication lines among the DGUs  \\
$\mathcal{E}_{\mathrm{eq}}$ & Set containing equilibrium point of the closed-loop dynamics\\
$\mathcal{N}$ & Set containing the electric buses in the EPDS \\
$I_{L}$ & Electrical load of the EPU \\
$I_{\tau,i}$ & The electric current generated by the $i$-th DGU \\
$L_{\tau,i}$ & Line inductance of the $i$-th power line \\
$R_{\tau,i}$ & Line resistance of the $i$-th power line \\
$u$ & Control law \\
$V_{\mathrm{dc}}$ & Voltage across PCC \\
$\xi_i$ & Controller states of the $i$-th  DGU\\
$w_i$ & Weight specified for the $i$-th  DGU \\
\bottomrule
\end{tabular}
\end{table}

\section{System Setup} \label{Sec:SHEPS}
In this section, we describe the architecture and the operational features of the \gls{epds} of an aircraft's series-hybrid-electric propulsion system. 

In Fig.~\ref{fig:GeneralHepArchitechture}, we show a schematic diagram with the main components of the considered propulsion system. The main power source is a \gls{tgs}, whereas an \gls{ess} acts as a secondary power source. The propulsion system is energized via two distinct lanes, each with its own \gls{epu}: the collective power of each source is injected into each lane's \gls{pcc}, which we will also refer to as the main DC load bus. Since, in nominal operating conditions, there is no power transfer between the lanes, in this work, we will focus exclusively on the operation of an individual lane and will neglect the effect that another lane may exert on it. Nonetheless, as we will see in Section 3, we generalize the system setup to accommodate for multiple energy sources.

With the objective of stating a number of simplifying modeling assumptions, below we provide further details about the main components of the propulsion system:
\begin{itemize}
    \item The \gls{tgs} comprises a gas turbine that drives a generator to produce AC electrical power. This AC power is converted to DC power using a rectifier. The \gls{tgs} is equipped with a low-level controller that accurately tracks desired voltage references, allowing the \gls{tgs} to serve as a voltage source.

    \item The \gls{ess} comprises a battery unit and a bi-directional DC-DC converter. Before a flight, the battery unit is charged, e.g., using renewable energy sources. Battery charging can also occur during a flight, using a scheme known as wind-milling, where the airflow over a free-spinning propeller generates electrical power. However, we assume that the battery operates exclusively in discharge mode all the time, with the converter functioning as a boost converter (see Remark~\ref{remark:charging}). Similar to the \gls{tgs}, the \gls{ess} is also equipped with a low-level controller that tracks desired voltage references. For more details on the voltage control of converter-interfaced sources, we refer to \cite{Morteza} and the references therein.

    {\color{black}
    \item The \gls{epu} consists of an inverter interfaced with a permanent magnet synchronous motor. The inverter converts the DC power drawn from the \gls{pcc} into AC power supplied to the motor. The mechanical power required by the propulsion system varies across flight segments (e.g., take-off, climb, and cruise). However, these variations occur on a much slower time scale than the electrical dynamics of the \gls{epds}. Therefore, over the time scale relevant for the voltage regulation problem considered in this work, the demanded power can be treated as constant. Moreover, the \gls{epds} is to be controlled to regulate the DC voltage at the \gls{pcc}. Under such regulated-voltage operation, the inverter of the \gls{epu} can operate in current-controlled mode. Consequently, since the PCC voltage is maintained approximately constant, a constant power demand can be approximated as a constant current drawn from the DC network.}

\end{itemize}

{\color{black}
\begin{remark} \label{remark:charging} The assumption that the battery operates only in discharge mode does not undermine the motivation for adopting a hybrid-electric propulsion architecture. In a series-hybrid-electric propulsion system, the battery fulfills several important roles beyond supplying average power. In particular, it provides peak power support during high-demand phases such as takeoff and climb, delivers fast transient response due to its superior power dynamics compared to the gas-turbine-driven generator, enhances redundancy and system reliability, and enables generator downsizing, thereby improving overall efficiency and weight distribution \cite{Benjamin}. Hence, the primary motivation for hybridization lies in performance, efficiency, and reliability considerations, rather than in the necessity of bidirectional battery operation during every flight phase.
    
\end{remark}
}

\begin{figure}[h!]
 \centering
  \includegraphics[width=0.45\textwidth]{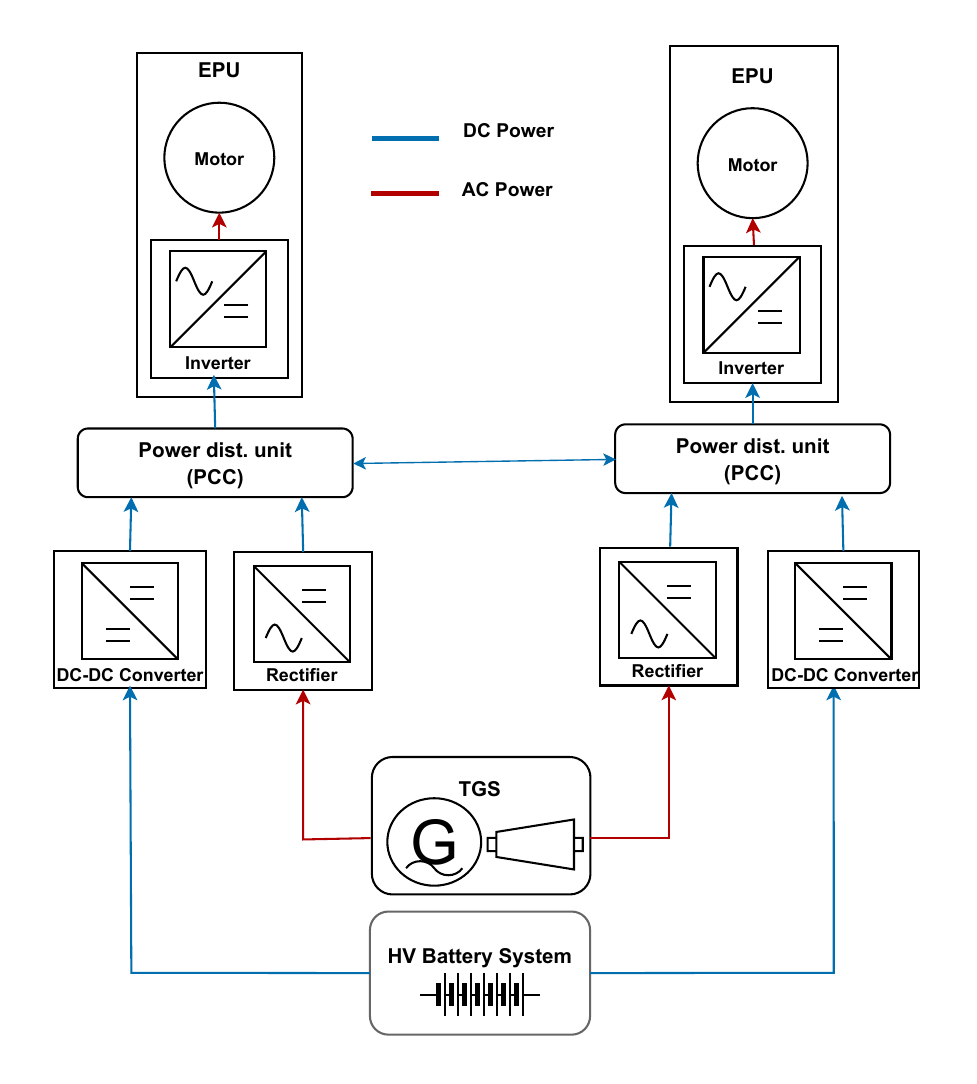}
 \caption{Exemplary series-hybrid-electric propulsion system.}
 \label{fig:GeneralHepArchitechture}
 \end{figure}

\section{MODEL, OBJECTIVES AND SOLUTION APPROACH} \label{Sec:modelAndPf}

\begin{figure*} 
\centering
\subfloat[Hybrid-electric propulsion system architecture. \label{fig:NominalHepArchitechture}]{\includegraphics[width=0.9\columnwidth]{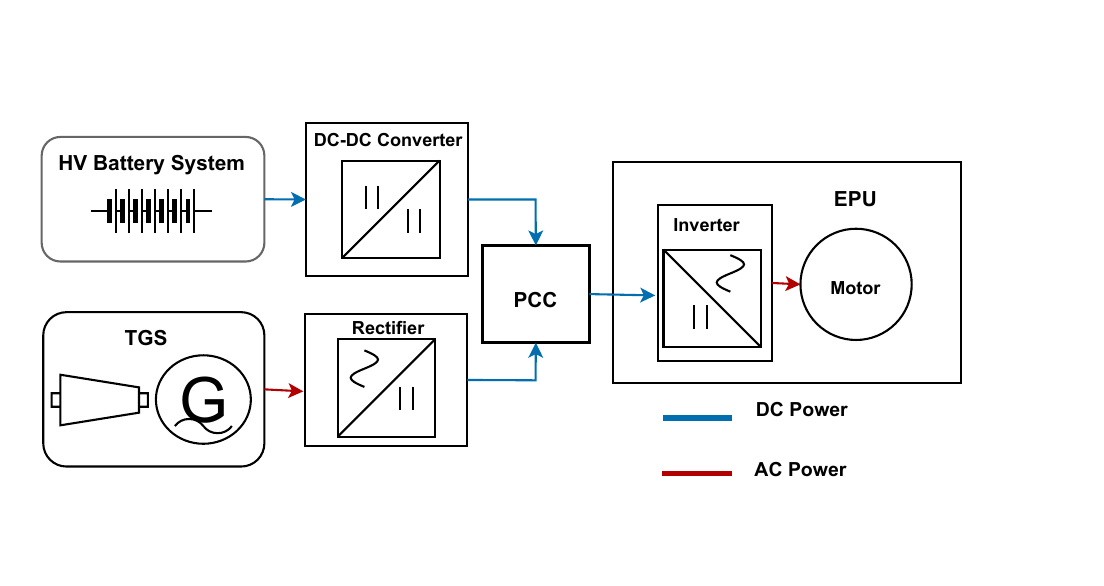}} 
\hfil
\subfloat[Electrical representation of a hybrid-electric propulsion system. \label{fig:DcMgHep}]{\includegraphics[width=0.9\columnwidth]{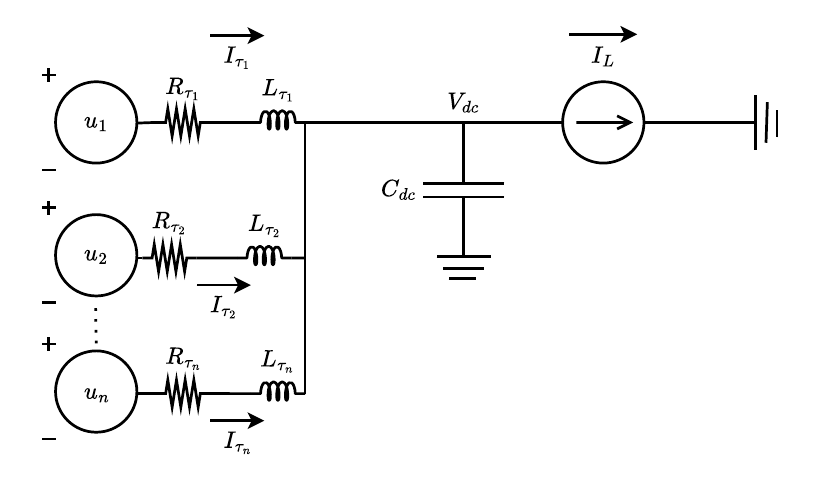}} 
\caption{Series-hybrid-electric propulsion system and its electrical representation.}
\label{fig:HepArchitectureAndSchematic}
\end{figure*}

Having described the system setup, in this section, we proceed to derive an ordinary differential equation-based system model. Furthermore, we outline the key control objectives and summarize the considered approach to address them.

\subsection{Modeling \gls{epds} of a Hybrid-Electric Propulsion System}

Let us consider a single lane of the propulsion system, as shown in Fig. \ref{fig:NominalHepArchitechture}. An equivalent electrical diagram of this lane is presented in Fig. \ref{fig:DcMgHep}. To account for the presence of multiple power sources, which are commonly used in certain aircraft propulsion systems (e.g., electric vertical take-off and landing (eVTOL)), we extend the propulsion architecture to a multi-source configuration to broaden the scope of this work. Since the \gls{epds} can be viewed as an onboard DC \gls{mg}, we adopt a modeling approach used in~\cite{Machado1,Noroozi2}. Accordingly, the \gls{epds} is represented as a connected graph \( \mathcal{G}(\mathcal{N}, \mathcal{E}) \), where the set of nodes \( \mathcal{N} = \{1, 2, \dots, n\} \) corresponds to electrical buses, and the set of edges \( \mathcal{E} = \{1, 2, \dots, n_\tau\} \) represents the power lines. Since we are considering that there is a unique load bus (from which the \gls{epu} is energized), we split the set of vertices $\mathcal{N}$ as $\mathcal{N}=\mathcal{N}_\mathrm{s}\cup \mathcal{N}_\ell$, where $\mathcal{N}_\mathrm{s}$ is comprised of the $n_\mathrm{s}> 1$ generation buses and $\mathcal{N}_\ell$ of the unique load bus; the number of power lines is denoted by $n_\tau$. 

Each \gls{dgu} (i.e., converter-interfaced generator or battery) is modeled as a controllable voltage source whose voltage $u_i(t)$ is treated as a control variable.  Each power line is modeled as a series connection of an inductor and a resistor \cite{Noroozi2}, \cite{Trip2}, \cite{Nahata}, \cite{Machado1}, \cite{Machado2}. The line inductance, resistance, and electric current are denoted by $L_{\tau,i}>0$,  $R_{\tau,i}>0$ and  $I_{\tau,i}(t)\in \mathbb{R}$, respectively. The load bus is modeled as a capacitor of capacitance $C_\mathrm{dc}>0$, and its voltage (w.r.t. ground) is denoted by $V_\mathrm{dc}(t)\in \mathbb{R}$. The \gls{epu}'s electrical load is denoted by $I_L(t) \in \mathbb{R}$.

Based on the above description, we use Kirchhoff's laws to obtain the dynamics of the current $I_{\tau,i}$ injected by the $i$-th \gls{dgu} as follows:

\begin{align}
	\begin{split}
		L_{\tau,i}	 \dot{I}_{\tau,i} &= -V_\mathrm{dc} - R_{\tau,i}I_{\tau,i} + u_i  \;  \;  \; \; \;  i= 1, 2, ..., n_\mathrm{s}. \\
  \label{eq:currentDynamics}
			\end{split}	 
\end{align}
Similarly, the dynamics of the voltage $V_\mathrm{dc}$ at the \gls{pcc} can be obtained as follows:
\begin{align}
	\begin{split}
		C_\mathrm{dc}	  \dot{V}_\mathrm{dc} &= \sum_{i=1}^{n_\mathrm{s}} I_{\tau,i} -I_L.\\	
  \label{eq:voltageDynamics}
	\end{split}	 
\end{align}
By combining \eqref{eq:currentDynamics} and \eqref{eq:voltageDynamics}, we obtain the model for the entire \gls{epds} in vector form as follows:
\begin{align}
	\begin{split}
		L_{\tau}	 \dot{I}_{\tau} &= -\boldsymbol{1}_{n_\mathrm{s}} V_\mathrm{dc} - R_{\tau}I_{\tau} + u, \\
		C_\mathrm{dc}	  \dot{V}_\mathrm{dc} &= \boldsymbol{1}_{n_\mathrm{s}}^\top I_{\tau} -I_L,\\	
		\label{eq:DcMgHepDiffEqn}
	\end{split}	 
\end{align}
where $L_{\tau}=\text{diag}(L_{\tau,i})$, $R_{\tau}=\text{diag}(R_{\tau,i})$, $I_{\tau}=\text{col}(I_{\tau,i})$, and $u=\text{col}(u_{i})$. 

{\color{black}
The following are modeling and measurement assumptions for the \gls{epds} model~\eqref{eq:DcMgHepDiffEqn}. 
}

\begin{assumption}\label{Assumption 1}  
~
\begin{enumerate}
	\item[\emph{(i)}] The internal dynamics of each DGU is an input-output stable subsystem capable of tracking desired reference signals. Consequently, the output DC voltage is our control signal $u_i$, for all $i\in \mathcal{N}_\mathrm{s}$ {\color{black} (c.f., \cite{BAI})}. 
 
	\item[\emph{(ii)}] For each DGU, the generated current $I_{\tau,i}$ is locally measured and available for control design purposes. Moreover, the voltage across the main DC bus (\gls{pcc}) $V_\mathrm{dc}$ is also measured, and its measurement is transmitted to each DGU. Thus, we define the output of each DGU as $y_i = \begin{bmatrix}I_{\tau,i} &  V_\mathrm{dc}\end{bmatrix}^\top$ {\color{black} (c.f., \cite{Guerrero, BAI, Nahata}).}

	\item[\emph{(iii)}] For each power line $i\in \mathcal{N}_\mathrm{s}$, the resistance $R_{\tau,i}>0$ is unknown, but constant (see Remark~\ref{remark:ContResistance}).

     \item[\emph{(iv)}] For each power line $i\in \mathcal{N}_\mathrm{s}$, the inductance $L_{\tau,i}>0$ is unknown, but belongs to the known open interval $]L_{\tau,i}^\text{min},L_{\tau,i}^\text{max}[$ (see Remark~\ref{remark:ContResistance}).

     \item[\emph{(v)}] The controlled \gls{epu} is modeled as a ZI-load, i.e., a parallel connection of a constant current load and a constant impedance, i.e., $I_L(V_{\mathrm{dc}}) = I_\ell + Y V_{\mathrm{dc}},$ where $I_\ell > 0$ and $Y > 0$ are {\color{black} unknown} but constant load parameters {\color{black} (c.f., \cite{Machado2}).} 
    
\end{enumerate}
\end{assumption}

{\color{black}
\begin{remark} \label{remark:ContResistance} The \gls{epds} model parameters are treated as constant. From a time-scale perspective, the closed-loop electrical dynamics evolve significantly faster than environmental variations caused by changes in altitude, pressure, and temperature. Therefore, these parameters can reasonably be assumed constant for control design purposes, which is standard practice in related literature; see, e.g.,~\cite{bastos,Ahmed, CheYanbo}. Moreover, although some parameters (e.g., line resistances) may vary with temperature, significant deviations from nominal values are not expected in hybrid-electric aircraft due to the use of thermal management systems that maintain components within prescribed operating ranges~\cite{Potamiti,Ouyang}. In addition, pressure variations associated with altitude changes are not expected to meaningfully affect electrical parameters, as sensitive components such as converters are typically placed in protected compartments designed to prevent environmental damage and humidity exposure~\cite{ASHRAE_Aircraft_2019}.

\end{remark}

\begin{remark}\label{remark:ZI_loadAssumption}
Assumption~\ref{Assumption 1}.(v) is consistent with the earlier discussion in the Section~\ref{Sec:SHEPS}  regarding the operation of the \gls{epu}. In particular, since the \gls{epds} is to be controlled to regulate the DC voltage at the \gls{pcc} and the propulsion power demand varies on a slower time scale than the electrical dynamics of the \gls{epds}, the inverter of the \gls{epu} can be operated in current-controlled mode. Consequently, the \gls{epu} can be modeled as drawing an approximately constant current from the DC network. In the model~\eqref{eq:DcMgHepDiffEqn}, this is represented as a constant current load $I_\ell$, while the internal losses of the \gls{epu} are captured by a shunt constant impedance ($1/Y$).
\end{remark}
}

\subsection{Objectives}

\noindent In this article, we pursue the following two objectives:
\begin{objective}[Voltage regulation] \label{obj:VoltReg}
\begin{equation}\label{eq:volt_reg}
\lim_{t \to \infty} \;  \; \; V_\mathrm{dc}(t) = V_\mathrm{dc}^*,    
\end{equation}
where $V_\mathrm{dc}^* \;  \in   \; \mathbb{R}$ is the desired voltage value at the PCC.

\end{objective}

\medskip

\begin{objective}[Current sharing] \label{obj:CurrShar}
\begin{equation}\label{eq:current_shar}
    \lim_{t \to \infty} \;  \; \; I_{\tau}(t) = \bar{I}_{\tau},
\end{equation}
where,
$$w_{i} \bar{I}_{\tau,i} = w_{j} \bar{I}_{\tau,j}, \;  \; \;  \; \text{ for all } \;  \; i, j \;  \; \epsilon \;  \; \mathcal{N}_\mathrm{s},$$ 
with $w_i>0$ being weights specified for each \gls{dgu}.  
\end{objective}

On the one hand, the achievement of Objective \ref{obj:VoltReg} will ensure the effective operation of the \gls{epu}, which has a specific operational voltage range, i.e., it is needed to accurately regulate the voltage at the \gls{pcc} to a prescribed nominal value. On the other hand, the achievement of Objective \ref{obj:CurrShar} will ensure that, at equilibrium, there exists a predefined proportionality relationship among the DGUs' injected currents, thereby contributing to preventing over-stressing them.

\bigskip

\subsection{Solution Approach}

\noindent
To summarize our approach, we introduce the assumption that there exists a connected communication graph $\mathcal{G}^\text{com}(\mathcal{N}_{\mathrm{s}}, \mathcal{E}^\text{com})$ among the \gls{dgu}s, with $\mathcal{N}_{\mathrm{s}}$ representing the set of all \gls{dgu}s and $\mathcal{E}^\text{com}$ the set of communication links among the \gls{dgu}s. Moreover, for any $i\in \mathcal{G}^\mathrm{com}$, we introduce the vector  $z_i$  grouping the output signals from all the \gls{dgu}s in $\mathcal{G}^\mathrm{com}$ which are adjacent to $i$.
{\color{black} Then, our approach is to design, for each input $u_i$, a dynamic and distributed feedback law of the form
\begin{subequations}\label{eq:control_generic_structure}
    \begin{align}
        u_i & = \gamma_i(y_i, z_i,\xi_i),\\
        \dot{\xi}_i & = \rho_i (y_i, z_i, \xi_i),
    \end{align}
\end{subequations}
where $\xi_i \in \mathbb{R}^{s_i}$ denotes the state of the $i$-th controller, $s_i$ is a natural number, and $\gamma_i$ and $\rho_i$ are continuous functions. The vector $y_i \in \mathbb{R}^2$ collects the measurable outputs of the $i$-th \gls{dgu}. The design of the controller is motivated by the fact that the model  \eqref{eq:DcMgHepDiffEqn} admits a port-Hamiltonian representation. This observation motivates the use of the Interconnection and Damping Assignment Passivity-Based Control (IDA-PBC) methodology with dynamic extension \cite{astolfi_ortega} for the controller synthesis (see Appendix for details). In addition, a backstepping-inspired change of variables is employed to introduce an integral action that compensates for unknown disturbances. Moreover, the feedback law \eqref{eq:control_generic_structure} aims to address parametric uncertainty. While the structural form of the \gls{epds} dynamics is assumed to be known, the exact parameter values need not be available for controller implementation. Instead, unknown parameters are compensated through the adaptive mechanism introduced in a subsequent section. Accordingly, the controller is designed such that the closed-loop dynamics obtained by interconnecting \eqref{eq:DcMgHepDiffEqn} with  \eqref{eq:control_generic_structure} admit a port-Hamiltonian structure with a stable equilibrium encoding the desired control objectives.
}

\section{A DISTRIBUTED CONTROL LAW FOR CURRENT SHARING AND VOLTAGE REGULATION} \label{Sec:Control}
In this section, we present a controller of the form \eqref{eq:control_generic_structure} tailored to meet Objectives~\ref{obj:VoltReg} and \ref{obj:CurrShar} for the model (\ref{eq:DcMgHepDiffEqn}) subject to Assumption~\ref{Assumption 1}. 
Subsequently, the achievement of Objectives 1 and 2 is implied from the proof that the trajectories of the resulting closed-loop system locally converge to an equilibrium point that encodes both voltage regulation and current sharing.

\subsection{Proposed Distributed Control Law} \label{subsec:PropsedControl}

Consider the system~\eqref{eq:DcMgHepDiffEqn}. Then, we propose to meet Objectives~\ref{obj:VoltReg} and \ref{obj:CurrShar} via the following distributed controller: 
\begin{equation}
\scalebox{0.8}{$
\begin{aligned}
T_{\varphi_{\tau,i}} \dot{\varphi}_{\tau,i} &= -(V_\mathrm{dc} - V^*_\mathrm{dc}) 
    - w_{i} \sum _{j \in \mathcal {N}^{com}_{i}}(\theta _{\tau,i} - \theta _{\tau,j}), \\
T_{\theta_{\tau,i}} \dot{\theta}_{\tau,i} &= \sum _{j \in \mathcal {N}^{com}_{i}}(w_{i} I_{\tau,i} - w_{j} I_{\tau,j}), \\
T_{\hat{r}_{\tau,i}} \dot{\hat{r}}_{\tau,i} &= -I_{\tau,i} (I_{\tau,i}-\varphi_{\tau,i}), \\
T_{\eta_{\tau,i}} \dot{\eta}_{\tau,i} &= -T^{-1}_{\varphi_{\tau,i}} 
    \left(-(V_\mathrm{dc} - V^*_\mathrm{dc}) - w_{i} \sum _{j \in \mathcal {N}^{com}_{i}}(\theta _{\tau,i} - \theta _{\tau,j})\right)(I_{\tau,i}-\varphi_{\tau,i}), \\
u_i &= -K_{z_{\tau,i}} (I_{\tau,i}-\varphi_{\tau,i}) + \hat{r}_{\tau,i} I_{\tau,i} + V^*_\mathrm{dc} \\
 &\quad + T^{-1}_{\varphi_{\tau,i}} 
    \left(-(V_\mathrm{dc} - V^*_\mathrm{dc}) - w_{i} \sum _{j \in \mathcal {N}^{com}_{i}}(\theta _{\tau,i} - \theta _{\tau,j}) \right) \eta_{\tau,i} \\
  &\quad - w_{i} \sum _{j \in \mathcal {N}^{com}_{i}}(\theta _{\tau,i} - \theta _{\tau,j}).
\end{aligned}$}
\label{eq:controllerInDistributed}
\end{equation}
where $K_{z_{\tau,i}}$, $T_{\theta_{\tau,i}}$, $T_{\varphi_{\tau,i}}$, $T_{\hat{r}_{\tau,i}}$, $T_{\eta_{\tau,i}}$ $\in \mathbb{R}$ are control parameters to tune the transient response, and $\xi_i= \text{col}(T_{\varphi_{\tau,i}}\varphi_{\tau,i}, T_{\theta_{\tau,i}}\theta_{\tau,i}, T_{\hat{r}_{\tau,i}}\hat{r}_{\tau,i}, T_{\eta_{\tau,i}}\eta_{\tau,i})$  are the controller's states. The set $\mathcal{N}^\text{com}_{i}$ is the set of \gls{dgu}s adjacent to the $i$-th \gls{dgu} within the communication network $\mathcal{G}^\text{com}$.

Note that for implementing the controller~\eqref{eq:controllerInDistributed}, it is necessary that the PCC's voltage measurement is communicated to all the \gls{dgu}s. Additionally,  measurements of the currents $I_{\tau,i}$ and the controller states $\theta_{\tau,i}$ need to be communicated among neighboring nodes within the communication network $\mathcal{G}^\text{com}$. This feature of the controller inherently makes it distributed in nature. 

{\color{black}A detailed explanation of the controller and the intuition behind its structure will be provided subsequently. At this stage, we find it convenient to represent the distributed controller in vector form for compactness and clarity.} The distributed controller \eqref{eq:controllerInDistributed} for all the \gls{dgu}s can be written in vector form as follows:
\begin{equation}
\scalebox{0.9}{$
\begin{aligned}
T_{\varphi_\tau} \dot{\varphi}_{\tau} &= -(V_\mathrm{dc} - V^*_\mathrm{dc})\boldsymbol{1}_{n_\mathrm{s}} - W \mathcal{L}^\text{com} \theta_{\tau}, \\
T_{\theta_\tau} \dot{\theta}_{\tau} &= \mathcal{L}^\text{com} W I_\tau, \\
T_{\hat{r}_\tau} \dot{\hat{r}}_{\tau} &= -\text{diag}(I_\tau) (I_\tau-\varphi_\tau), \\
T_{\eta_{\tau}} \dot{\eta}_{\tau} &= -\text{diag} \left( T^{-1}_{\varphi_\tau} (-(V_\mathrm{dc} - V^*_\mathrm{dc})\boldsymbol{1}_{n_\mathrm{s}} - W \mathcal{L}^\text{com} \theta_{\tau}) \right)  (I_\tau-\varphi_\tau), \\
u &= -K_{z_\tau} (I_\tau-\varphi_\tau) + \text{diag}(\hat{r}_{\tau}) I_\tau + \boldsymbol{1}_{n_\mathrm{s}}V^*_\mathrm{dc} \\
 &\quad + \text{diag} \left( T^{-1}_{\varphi_\tau} (-(V_\mathrm{dc} - V^*_\mathrm{dc})\boldsymbol{1}_{n_\mathrm{s}} - W \mathcal{L}^\text{com} \theta_{\tau}) \right) \eta_{\tau} \\
  &\quad -W \mathcal{L}^\text{com} \theta_{\tau},
\end{aligned}$}
\label{eq:controller}
\end{equation}
{\color{black} where $K_{z_\tau} = \text{diag}(K_{z_{\tau,i}})$ $\in \mathbb{R}^{n_\mathrm{s} \times n_\mathrm{s}}$, $T_{\theta_\tau} = \text{diag}(T_{\theta_{\tau,i}})$ $\in \mathbb{R}^{n_\mathrm{s} \times n_\mathrm{s}}$, $T_{\varphi_\tau} = \text{diag}(T_{\varphi_{\tau,i}})$ $\in \mathbb{R}^{n_\mathrm{s} \times n_\mathrm{s}}$, $T_{\hat{r}_\tau} = \text{diag}(T_{\hat{r}_{\tau,i}})$ $\in \mathbb{R}^{n_\mathrm{s} \times n_\mathrm{s}}$, $T_{\eta_\tau} = \text{diag}(T_{\eta_{\tau,i}})$ $\in \mathbb{R}^{n_\mathrm{s} \times n_\mathrm{s}}$ are positive-definite diagonal matrices of control parameters, and $\mathcal{L}^\text{com}\in \mathbb{R}^{n_\mathrm{s}\times n_\mathrm{s}}$ is the Laplacian matrix of the communication graph $\mathcal{G}^\mathrm{com}$. Additionally, $W=\text{diag}(w_i)$ $\in \mathbb{R}^{n_\mathrm{s} \times n_\mathrm{s}}$ denotes a diagonal weighting matrix.}

{\color{black} The controller structure incorporates the additional state variables ($\varphi_\tau$, $\theta_\tau$, $\hat{r}_\tau$, $\eta_\tau$) to achieve the desired objectives. First, the dynamic extension adopted from \cite{Nahata}, introduced through the variable $\varphi_\tau$, plays a role analogous to the integral action used in classical output regulation. In particular, it accumulates the voltage regulation error and ensures that, at steady-state, $V_{\mathrm{dc}} = V^*_{\mathrm{dc}}$ (Objective~\ref{obj:VoltReg}), thereby eliminating steady-state voltage deviation. Then, the consensus-driven term $\mathcal{L}^\mathrm{com} W I_\tau$, adopted from \cite{Nahata, Trip2}, is included in the $\dot{\theta}_\tau$ dynamics to enforce weighted proportional current sharing among the DGUs (Objective~\ref{obj:CurrShar}). Then, the variables $\hat{r}_\tau$ and $\eta_\tau$ are introduced to compensate for the uncertainty in the resistances and inductances of the lines. The controller does not require prior knowledge of these parameters; instead, adaptation mechanisms are employed to ensure voltage regulation and proportional current-sharing despite parametric uncertainty. In particular, $\hat{r}_\tau$ is introduced as an estimate of the unknown line resistances that converge to the actual line resistance at steady-state, and its update law follows the gradient-based adaptive structure \cite{slotine}. Some additional benefits of resistance estimation are highlighted in the Remark~\ref{remark:Res_est_ben}. 
In contrast, $\eta_\tau$ does not necessarily converge to the true line inductance. Rather, it evolves toward a bounded constant value that is sufficient to guarantee the desired closed-loop behavior. The role of $\eta_\tau$ is not parameter identification but dynamic compensation of uncertain inductive effects while preserving the passivity structure of the closed-loop system, which is essential for ensuring closed-loop stability. As shown in Section~\ref{subsec:EqAnalysis}, steady-state objective satisfaction is achieved without requiring convergence of $\eta_\tau$ to the physical inductance, while Section~\ref{subsec:stabilityAnalysis} establishes that such convergence is likewise unnecessary for stability.

Finally, the control law for  $u$ is designed such that the desired closed-loop equilibrium is stable.  For that, we have followed the nonlinear control design methodology known as Interconnection and Damping Assignment Passivity-Based Control (IDA-PBC) with dynamic extension \cite{astolfi_ortega}. This controller renders the plant's closed-loop dynamics a port-Hamiltonian system \cite{arjan_book_port_Ham} with a stable or asymptotically stable equilibrium. Details of the methodology are provided in the appendix. For further reading on port-Hamiltonian systems and IDA-PBC design, interested readers are referred to \cite{arjan_book_l2S_3ed} and \cite{canseco_ortega_IDA_PBC}.
}

\bigskip

{\color{black}
\begin{remark} \label{remark:Res_est_ben} Knowledge of the estimated resistances provides additional benefits beyond improving voltage regulation and current-sharing performance. In particular, they enable the computation of resistive power losses, which can support supervisory monitoring of network efficiency and operating conditions. Moreover, deviations of the estimates from nominal values may indicate abnormal conditions such as line degradation, loose connections, or developing faults. Early detection of such anomalies is particularly important in aircraft applications due to stringent reliability and safety requirements. Finally, the availability of these estimates can facilitate future extensions toward optimization-based power flow coordination that explicitly accounts for network losses.
    
\end{remark}
}

\subsection{Closed-loop Dynamics}

\noindent Moving forward, we introduce a new variable $z_\tau$ as follows:   

\begin{align}
	z_{\tau} := {I}_{\tau} - \varphi_{\tau},
	\label{eq:NewVariable}
\end{align}
where $\varphi_{\tau}$ evolves according to \eqref{eq:controller}. {\color{black} This change of variable—motivated by backstepping control design \cite{khalil2}—introduces an integral action in the $V_\mathrm{dc}$ dynamics, which is essential for compensating the unknown constant disturbance $I_\ell$. The change of variable is then employed in the subsequent proposition to derive an equivalent representation of the closed-loop system obtained by combining \eqref{eq:DcMgHepDiffEqn} and \eqref{eq:controller}.}  

\begin{proposition}
Consider the aircraft EPDS model \eqref{eq:DcMgHepDiffEqn}, subject to Assumption 1 and in closed-loop with the distributed controller \eqref{eq:controller}. Consider the change of variables \eqref{eq:NewVariable} and define the following vectors:
\begin{subequations}
    \begin{align}
         r_{\tau} &  := R_{\tau}\boldsymbol{1}_{n_\mathrm{s}},\label{eq:definition_r_tau}\\
          l_{\tau} & := L_{\tau}\boldsymbol{1}_{n_\mathrm{s}}.
          \label{eq:definition_l_tau}
    \end{align}
\end{subequations}
Then, the closed-loop system can be equivalently represented as follows:
\begin{equation}
\scalebox{0.9}{$
\begin{aligned}
L_{\tau} \dot{z}_{\tau} &= -K_{z_\tau} z_{\tau} + \text{diag} (\varphi_{\tau} + z_{\tau}) (\hat{r}_{\tau}-r_{\tau}) \\
		& \; \; \; -(V_\mathrm{dc} - V^*_\mathrm{dc})\boldsymbol{1}_{n_\mathrm{s}} - W \mathcal{L}^\text{com} \theta_{\tau} \\
         & \; \; \;  + \text{diag} \left( T^{-1}_{\varphi_\tau} (-(V_\mathrm{dc} - V^*_\mathrm{dc})\boldsymbol{1}_{n_\mathrm{s}} - W  \mathcal{L}^\text{com} \theta_{\tau}) \right) (\eta_{\tau} - l_{\tau}), \\
		C_\mathrm{dc} \dot{V}_\mathrm{dc} &= \boldsymbol{1}_{n_\mathrm{s}}^\top (z_{\tau} + \varphi_{\tau}) -I_{\ell} -  Y V_\mathrm{dc},\\
		T_{\varphi_\tau} \dot{\varphi}_{\tau} &= -(V_\mathrm{dc} - V^*_\mathrm{dc})\boldsymbol{1}_{n_\mathrm{s}} - W \mathcal{L}^\text{com} \theta_{\tau}, \\
		T_{\theta_\tau} \dot{\theta}_{\tau} &= \mathcal{L}^\text{com} W (z_{\tau} + \varphi_{\tau}), \\
		T_{\hat{r}_\tau} \dot{\hat{r}}_{\tau} &= -\text{diag} (\varphi_{\tau} + z_{\tau})z_{\tau} \\
        T_{\eta_{\tau}} \dot{\eta}_{\tau} &= -\text{diag} \left( T^{-1}_{\varphi_\tau} (-(V_\mathrm{dc} - V^*_\mathrm{dc})\boldsymbol{1}_{n_\mathrm{s}} - W \mathcal{L}^\text{com} \theta_{\tau}) \right)  z_{\tau}.
\end{aligned}$}
\label{eq:closed-loop}
\end{equation}

\end{proposition}
\begin{proof}
In view of  \eqref{eq:NewVariable}  the time-derivative of $z_\tau$ satisfies:
\begin{align}
	\begin{split}
		L_{\tau} \dot{z}_{\tau} &= L_{\tau} \dot{I}_{\tau} - L_{\tau} \dot{\varphi}_{\tau}.
	\end{split}
\label{eq:Lz1}
\end{align}
By substituting the expression for $L_{\tau} \dot{I}_{\tau}$ from (\ref{eq:DcMgHepDiffEqn}) and $\dot{\varphi}_{\tau}$ from (\ref{eq:controller}) in (\ref{eq:Lz1}), we get:
\begin{align}
	\begin{split}
		L_{\tau} \dot{z}_{\tau} &= (-\boldsymbol{1}_{n_\mathrm{s}} V_\mathrm{dc} - R_{\tau}I_{\tau} + u) - \\
		& \; \; \; L_{\tau}T_{\varphi_\tau}^{-1}(-(V_\mathrm{dc} - V^*_\mathrm{dc})\boldsymbol{1}_{n_\mathrm{s}} - W \mathcal{L}^\text{com} \theta_{\tau}).
	\end{split}
	\label{eq:Lz2}
\end{align}
By substituting $u$ from (\ref{eq:controller}) in (\ref{eq:Lz2}) we obtain:
\begin{align}
\begin{split}
		L_{\tau} \dot{z}_{\tau} &= -\boldsymbol{1}_{n_\mathrm{s}}(V_\mathrm{dc} - V_\mathrm{dc}^*) -  R_{\tau}(z_{\tau} + \varphi_{\tau}) \\
		& \; \; \; -K_{z_\tau} z_{\tau} + \text{diag}(\hat{r}_{\tau}) (z_{\tau} + \varphi_{\tau}) -W \mathcal{L}^\text{com} \theta_{\tau} \\ 
		& \; \; \;  + \text{diag} \left( T^{-1}_{\varphi_\tau} (-(V_\mathrm{dc} - V^*_\mathrm{dc})\boldsymbol{1}_{n_\mathrm{s}} - W \mathcal{L}^\text{com} \theta_{\tau}) \right) \eta_{\tau} \\
        & \; \; \; -  L_{\tau} \left( T^{-1}_{\varphi_\tau} (-(V_\mathrm{dc} - V^*_\mathrm{dc})\boldsymbol{1}_{n_\mathrm{s}} - W \mathcal{L}^\text{com} \theta_{\tau}) \right).
	\end{split}
	\label{eq:Lz3}
\end{align}
By considering the identities
\begin{subequations}
    \begin{align}
        &   (\text{diag}(\hat{r}_{\tau})-R_{\tau})(z_{\tau} + \varphi_{\tau}) = \text{diag}(z_{\tau} + \varphi_{\tau})(\hat{r}_{\tau}-R_{\tau}\boldsymbol{1}_{n_\mathrm{s}}),\\
& \text{diag} \left(T^{-1}_{\varphi_\tau} \left( -(V_\mathrm{dc} - V^*_\mathrm{dc})\boldsymbol{1}_{n_\mathrm{s}} - W \mathcal{L}^\text{com} \theta_{\tau} \right) \right) {l}_{\tau} \nonumber \\
&= L_{\tau} T^{-1}_{\varphi_\tau} \left( -(V_\mathrm{dc} - V^*_\mathrm{dc})\boldsymbol{1}_{n_\mathrm{s}} - W \mathcal{L}^\text{com} \theta_{\tau} \right), \label{eq:l_Identity}
    \end{align}
\end{subequations}
then \eqref{eq:Lz3} can be equivalently written as: 
\begin{equation}
\scalebox{0.9}{$
\begin{aligned}
L_{\tau} \dot{z}_{\tau} &= -K_{z_\tau} z_{\tau} + \text{diag} (\varphi_{\tau} + z_{\tau}) (\hat{r}_{\tau}-r_{\tau}) \\
		& \; \; \; -(V_\mathrm{dc} - V^*_\mathrm{dc})\boldsymbol{1}_{n_\mathrm{s}} - W \mathcal{L}^\text{com} \theta_{\tau} \\
         & \; \; \; + \text{diag} \left( T^{-1}_{\varphi_\tau} (-(V_\mathrm{dc} - V^*_\mathrm{dc})\boldsymbol{1}_{n_\mathrm{s}} - W  \mathcal{L}^\text{com} \theta_{\tau}) \right) (\eta_{\tau} - l_{\tau}). \\
\end{aligned}$}
\label{eq:closed-loop-proved}
\end{equation}
By incorporating the $V_\mathrm{dc}$-dynamics, $z_{\tau}$-dynamics, and the controller from the model \eqref{eq:DcMgHepDiffEqn}, \eqref{eq:closed-loop-proved}, and \eqref{eq:controller}, respectively, and propagating the change of variable \eqref{eq:NewVariable}, we can establish that the closed-loop system is equivalent to \eqref{eq:closed-loop}. $\hfill\blacksquare$
\end{proof} 

{\color{black}
\begin{remark}\label{rem:invariant_set_theta}
Consider the closed-loop system \eqref{eq:closed-loop}. Multiplying both sides of the $\theta_{\tau}$-dynamics  from the left by $\boldsymbol{1}_{n_\mathrm{s}}^\top$ we get that $\boldsymbol{1}_{n_\mathrm{s}}^\top T_{\theta_\tau}\dot{\theta}_\tau=0$ for all time. Then, the  solutions of \eqref{eq:closed-loop} evolve  in the following invariant set:
\begin{equation}
\scalebox{1.0}{$
\begin{aligned}
\mathcal{I}=\{(z_\tau,V_\mathrm{dc},\varphi_\tau,\theta_\tau,\hat{r}_\tau,\eta_{\tau}) \in \mathbb{R}^{5n_\mathrm{s}+1}:~\boldsymbol{1}_{n_\mathrm{s}}^\top T_{\theta_\tau}\theta_\tau =c\},  
\end{aligned}$}
\label{eq:invariant_set_theta}
\end{equation}
\noindent where $c=\boldsymbol{1}_{n_\mathrm{s}}^\top T_{\theta_\tau}\theta_\tau(0)$. This set is a hyperplane in the full state space that constrains only the weighted sum of $\theta_\tau$, while all remaining state variables are unconstrained. Consequently, the weighted average of $\theta_\tau$ is preserved by the closed-loop dynamics, and the trajectories remain in the corresponding level set for all time. This fact is key in showing later the convergence to an equilibrium. Importantly, this dependence on $\theta_\tau(0)$ does not affect the control objectives. As will be shown later, the equilibrium values of $V_{\mathrm{dc}}$ and $I_\tau$, which encode voltage regulation and proportional current sharing, are independent of $\theta_\tau(0)$. The initial condition only determines a constant offset of $\theta_\tau$ within the invariant set and does not restrict feasibility nor degrade closed-loop performance. 
\end{remark}
}

\subsection{Closed-loop System Equilibria Analysis} \label{subsec:EqAnalysis}

The following proposition establishes that, for any given $\theta_\tau(0)$, the closed-loop system~\eqref{eq:closed-loop}  admits an equilibrium set apt to meet  Objectives~\ref{obj:VoltReg} and \ref{obj:CurrShar}. Before proceeding further, we find it convenient to represent the closed-loop system's state vector in the following form:
\begin{equation}\label{eq:state_vector}
    x=(L_\tau z_\tau,C_\mathrm{dc}V_\mathrm{dc},T_{\varphi_\tau}\varphi_\tau,T_{\theta_\tau}\theta_\tau,T_{\hat{r}_\tau}\hat{r}_\tau,T_{\eta_{\tau}}\eta_{\tau}).
\end{equation}

\begin{proposition} \label{prop:eqAchieveObj}
Subject to Assumption~\ref{Assumption 1}, for any given $\theta_{\tau}(0) \in \mathbb{R}^{n_\mathrm{s}}$, the closed-loop dynamics~\eqref{eq:closed-loop} admits a non-empty equilibrium set $\mathcal{E}_{\mathrm{eq}}$, every point   of which is of the form
\begin{equation}
(\boldsymbol{0}_{n_\mathrm{s}},C_\mathrm{dc}V_\mathrm{dc}^*,\alpha T_{\varphi_\tau} W^{-1}\boldsymbol{1}_{n_\mathrm{s}},\beta T_{\theta_\tau} \boldsymbol{1}_{n_\mathrm{s}},T_{\hat{r}_\tau}r_\tau, T_{\eta_{\tau}} \bar{\eta}_{\tau}),
\label{eq:general_equilibrium}
\end{equation}
where,
\begin{equation}\label{eq:values_alpha_beta}
        \alpha  = \frac{I_\ell +Y{V}_{\mathrm{dc}}^*}{\boldsymbol{1}_{n_\mathrm{s}}^\top W^{-1}\boldsymbol{1}_{n_\mathrm{s}}},~\quad\\
        \beta  =\frac{\boldsymbol{1}_{n_\mathrm{s}}^\top T_{\theta_\tau} \theta_{\tau}(0)}{\boldsymbol{1}_{n_\mathrm{s}}^\top T_{\theta_\tau}\boldsymbol{1}_{n_\mathrm{s}}},
\end{equation}
and $\bar{\eta}_{\tau} \in \mathbb{R}^{n_\mathrm{s}}$ is any constant vector.
Moreover, the vector of equilibrium currents, $\bar I_\tau$,  is such that
\begin{equation}\label{subeq:ssCurrentSharing}
    w_{i} \bar{I}_{\tau,i} = w_{j} \bar{I}_{\tau,j}, \;  \; \;  \; \text{ for all } \;  \; i, j \;  \; \epsilon \;  \; \mathcal{N}_\mathrm{s}.
\end{equation}
\end{proposition}

\begin{proof}
Consider \eqref{eq:closed-loop} at steady-state.
The equilibrium condition $\dot{\varphi}_{\tau} = 0$ holds if and only if $-(\bar V_\mathrm{dc}-V_\mathrm{dc}^*)\boldsymbol{1}_{n_\mathrm{s}}-W\mathcal{L}^\mathrm{com}\bar \theta_\tau=0$. By multiplying both sides of this equation by $\boldsymbol{1}_{n_\mathrm{s}}W^{-1}$ we get that $(\bar{V}_\mathrm{dc}-V_\mathrm{dc}^*)\boldsymbol{1}_{n_\mathrm{s}}^\top W^{-1}\boldsymbol{1}_{n_\mathrm{s}}=0$. Since $W$ is a diagonal, positive-definite matrix, it follows that $\bar V_\mathrm{dc}=V_\mathrm{dc}^*$.

Now,  $\dot{\theta}_{\tau} = 0$ holds if and only if 
\begin{equation}\label{eq:proof_ss_1}
    W (\bar{z}_\tau+\bar{\varphi}_\tau) = a \boldsymbol{1}_{n_\mathrm{s}},
\end{equation}
for some $a\in \mathbb{R}$. To verify that $a=\alpha$, with $\alpha$ defined in \eqref{eq:values_alpha_beta}, we use the equilibrium condition $\dot{V}_\mathrm{dc}=0$, from which we get that
\begin{equation}\label{eq:proof_ss_2}
    \boldsymbol{1}_{n_\mathrm{s}}^\top (\bar{z}_\tau+\bar{\varphi}_\tau)  = I_\ell+Y\bar{V}_\mathrm{dc}.
    \end{equation}
Clearing $\bar{z}_\tau+\bar{\varphi}_\tau$ from \eqref{eq:proof_ss_1} and substituting it into \eqref{eq:proof_ss_2} results in
\begin{align}\label{eq:proof_ss_2.1}
    \boldsymbol{1}_{n_\mathrm{s}}^\top \left(a W^{-1}\boldsymbol{1}_{n_\mathrm{s}} \right) & = I_\ell+Y\bar{V}_\mathrm{dc} \nonumber\\
    & \Leftrightarrow \nonumber\\
    a & =\frac{I_\ell +Y\bar{V}_{\mathrm{dc}}}{\boldsymbol{1}_{n_\mathrm{s}}^\top W^{-1}\boldsymbol{1}_{n_\mathrm{s}}}>0,
\end{align}
we note that $a>0$ due to Assumption~\ref{Assumption 1}.(v).

Observe now from \eqref{eq:proof_ss_1} and \eqref{eq:proof_ss_2.1} that $\bar{z}_{\tau,i}+\bar{\varphi}_{\tau,i}>0$ for all $i\in \mathcal{N}_\mathrm{s}$. Then, the equilibrium condition $\dot{\hat{r}}_{\tau}=0$ in \eqref{eq:closed-loop} holds if and only if $\bar{z}_\tau=\boldsymbol{0}_{n_\mathrm{s}}$. Having established this, we can see from \eqref{eq:proof_ss_1} that $\bar{\varphi}_\tau$ satisfies  \eqref{eq:general_equilibrium}. 

Now let us verify  that $\bar{\hat{r}}_\tau=r_\tau$. Simultaneously considering the equilibrium conditions $\dot{z}_\tau=0$ and $\dot{\varphi}_\tau=0$ it is possible to obtain the following identity:
\begin{equation}\label{eq:proof_ss_3}
    0 = -K_{z_\tau}\bar{z}_\tau + \text{diag}(\bar{z}_\tau +\bar{\varphi}_\tau)(\bar{\hat{r}}_\tau-r_\tau).
\end{equation}
Since $\bar{z}_\tau=\boldsymbol{0}_{n_\mathrm{s}}$ and $\bar{\varphi}_{\tau,i}>0$ for all $i\in \mathcal{N}_\mathrm{s}$, then \eqref{eq:proof_ss_3} holds if and only if $\bar{\hat{r}}_\tau = r_\tau$.

Moving forward, consider again the equilibrium condition $\dot{\varphi}_\tau=0$. Having  established that $\bar{V}_\mathrm{dc}=V_\mathrm{dc}^*$, it follows that $\dot{\varphi}_\tau=0$ if and only if
\begin{equation}\label{eq:proof_ss_4}
    0  = -W\mathcal{L}^\mathrm{com}\bar{\theta}_\tau  \Leftrightarrow 
    \bar{\theta}_\tau  = b \boldsymbol{1}_{n_\mathrm{s}},
\end{equation}
for some $b\in \mathbb{R}$. In view of Remark~\ref{rem:invariant_set_theta}, $\bar{\theta}_\tau$ satisfying \eqref{eq:proof_ss_4} is a feasible equilibrium value if and only if $\boldsymbol{1}_{n_\mathrm{s}}^\top T_{\theta_\tau} \bar{\theta}_\tau = \boldsymbol{1}_{n_\mathrm{s}}^\top T_{\theta_\tau}\theta_\tau(0) $, which holds if $b=\beta$, where $\beta$ is given in \eqref{eq:values_alpha_beta}. 

Having proved  that $\bar z_\tau$, $\bar V_\mathrm{dc}$, $\bar \varphi_\tau$, $\bar \theta_\tau$ and $\bar{\hat{r}}_\tau$ are as in \eqref{eq:general_equilibrium}, it is direct to see that   \( \dot{\eta}_{\tau} = 0 \) and \( \dot{z}_{\tau} = 0 \) hold for any constant \( \bar{\eta}_{\tau} \), as stated. Finally, the satisfaction of \eqref{subeq:ssCurrentSharing} directly follows from \eqref{eq:NewVariable} and \eqref{eq:proof_ss_1}, and from the fact that $\bar{z}_\tau =\boldsymbol{0}_{n_\mathrm{s}}$.  $\hfill\blacksquare$
\end{proof}

\subsection{Closed-loop System Stability and Convergence Analysis} \label{subsec:stabilityAnalysis}

The following proposition establishes the technical conditions under which Objectives~\ref{obj:VoltReg} and~\ref{obj:CurrShar} are satisfied. On the one hand, note from \eqref{eq:general_equilibrium} that any closed-loop equilibrium satisfies $\bar{\hat{r}}_\tau=r_\tau$, which implies that the proposed controller offers the possibility of accurately estimating the line resistance matrix $R_\tau$, whose main diagonal entries have been arranged in the vector $r_\tau$ (see \eqref{eq:definition_r_tau}). This feature is central for the satisfaction of Objective~\ref{obj:VoltReg}. On the other hand, note that $\bar \eta$ is not necessarily equal to $l_\tau$; in fact $\bar \eta$ can be any constant vector. We show in the sequel that this is not needed for the satisfaction of Objectives ~\ref{obj:VoltReg} and ~\ref{obj:CurrShar} (see, e.g., \cite[Example 4.10]{khalil2}).

\begin{proposition}\label{prop:local_stab}
Consider the closed-loop system \eqref{eq:closed-loop} with Assumption~\ref{Assumption 1}. Choose
 \begin{align} \label{eq:Tphi_conditions}
T_{\varphi_{\tau,i}} &> L_{\tau,i}^{\mathrm{max}} - L_{\tau,i}^{\mathrm{min}},\quad i=1,2,\ldots,n_s.
\end{align}
Then, all solutions starting sufficiently close to the equilibrium point
\begin{equation}\label{eq:desired_equilibrium}
\bar x=(\boldsymbol{0}_{n_\mathrm{s}},C_\mathrm{dc}V_\mathrm{dc}^*,\alpha T_{\varphi_\tau} W^{-1}\boldsymbol{1}_{n_\mathrm{s}},  \beta'  T_{\theta_\tau} \boldsymbol{1}_{n_\mathrm{s}},T_{\hat{r}_\tau}r_\tau, T_{\eta_{\tau}}l_\tau),
\end{equation}
with $\alpha$ as in \eqref{eq:values_alpha_beta}  and $\beta' \in \mathbb{R}$ being an arbitrary constant, converge to the equilibrium set \(\mathcal{E}_{\mathrm{eq}}\) characterized by \eqref{eq:general_equilibrium}.
Moreover, \eqref{eq:volt_reg} and \eqref{eq:current_shar} hold.
\end{proposition}

\begin{proof}
Consider the following Lyapunov function candidate for $\bar x$:
\begin{align}
	\begin{split}
	S(x,\bar x) &= \frac{1}{2} \|z_{\tau}\|^2_{L_\tau} + \frac{1}{2} \|V_\mathrm{dc} - \bar{V}_\mathrm{dc}\|^2_{C_\mathrm{dc}}    \\ 
 & \; \; \; + \frac{1}{2} \|\varphi_{\tau} - \bar{\varphi}_{\tau}\|^2_{T_{\varphi_\tau}} 
	  + \frac{1}{2} \|\theta_{\tau} - \bar{\theta}_{\tau}\|^2_{T_{\theta_\tau}} \\
  & \; \; \; + \frac{1}{2} \|\hat{r}_{\tau} - \bar{\hat{r}}_{\tau}\|^2_{T_{\hat{r}_\tau}} + \frac{1}{2} \|\eta_{\tau} - l_{\tau}\|^2_{T_{\eta_{\tau}}}.	
	\end{split}
    \label{lyap}
\end{align}
$S$ is globally positive-definite with respect to  $\bar x$ and its time-derivative along trajectories of  \eqref{eq:closed-loop} is given by
\begin{equation}
\scalebox{1.0}{$
\begin{aligned}
	\dot{S}(x,\bar x) &= z^\top_\tau L_\tau \dot{z}_\tau + (V_\mathrm{dc} - \bar{V}_\mathrm{dc})^\top C_\mathrm{dc} \dot{V}_\mathrm{dc}  \\
	&\quad + (\varphi_{\tau} - \bar{\varphi}_{\tau})^\top T_{\varphi_\tau} \dot{\varphi}_\tau + (\theta_{\tau} - \bar{\theta}_{\tau}) T_{\theta_\tau} \dot{\theta}_\tau \\
    &\quad + (\hat{r}_{\tau} - \bar{\hat{r}}_{\tau})^\top T_{\hat{r}_\tau} \dot{\hat{r}}_{\tau} + (\eta_{\tau} - l_{\tau})^\top T_{\eta_\tau} \dot{\eta}_{\tau}.
\end{aligned}
$}
\label{eq:sdot}
\end{equation}
To further develop \eqref{eq:sdot}, we shift the right-hand side vector field of  \eqref{eq:closed-loop} with respect  to its value at $\bar x$, obtaining that
\begin{equation}
\scalebox{0.9}{$
\begin{aligned}
L_{\tau} \dot{z}_{\tau} &= -K_{z_\tau} z_{\tau} + \text{diag} (\varphi_{\tau} + z_{\tau}) (\hat{r}_{\tau}-\bar{\hat{r}}_{\tau}) \\
		& \; \; \; -(V_\mathrm{dc} - \bar{V}_\mathrm{dc})\boldsymbol{1}_{n_\mathrm{s}} - W \mathcal{L}^\text{com}(\theta_{\tau} - \bar{\theta}_{\tau}) \\
        & \; \; \; + \text{diag} \left( T^{-1}_{\varphi_\tau} (-(V_\mathrm{dc} - V^*_\mathrm{dc})\boldsymbol{1}_{n_\mathrm{s}} - W  \mathcal{L}^\text{com} \theta_{\tau}) \right) (\eta_{\tau} - l_{\tau}), \\
		C_\mathrm{dc} \dot{V}_\mathrm{dc} &= \boldsymbol{1}_{n_\mathrm{s}}^\top z_{\tau} + \boldsymbol{1}_{n_\mathrm{s}}^\top (\varphi_{\tau} - \bar{\varphi}_{\tau}) -  Y(V_\mathrm{dc} - \bar{V}_\mathrm{dc}),\\  
		T_{\varphi_\tau} \dot{\varphi}_{\tau} &= -(V_\mathrm{dc} - \bar{V}_\mathrm{dc})\boldsymbol{1}_{n_\mathrm{s}} - W \mathcal{L}^\text{com}(\theta_{\tau}-\bar{\theta}_{\tau}), \\
		T_{\theta_\tau} \dot{\theta}_{\tau} &= \mathcal{L}^\text{com} W z_{\tau} + \mathcal{L}^\text{com} W(\varphi_{\tau}-\bar{\varphi}_{\tau}), \\
		T_{\hat{r}_\tau} \dot{\hat{r}}_{\tau} &= -\text{diag} (\varphi_{\tau} + z_{\tau})z_{\tau},\\ 
         T_{\eta_{\tau}} \dot{\eta}_{\tau} &= -\text{diag} \left( T^{-1}_{\varphi_\tau} (-(V_\mathrm{dc} - V^*_\mathrm{dc})\boldsymbol{1}_{n_\mathrm{s}} - W \mathcal{L}^\text{com} \theta_{\tau}) \right)  z_{\tau}.
\end{aligned}$}
\label{eq:shifted_closed-loop}
\end{equation}
Then,  $\dot S$ satisfies the identity
\begin{align}
	\begin{split}
		\dot{S}(x,\bar x) &= -z^\top_\tau	K_{z_\tau} z_\tau -  (V_\mathrm{dc} - \bar{V}_\mathrm{dc})^\top Y(V_\mathrm{dc} - \bar{V}_\mathrm{dc}).
	\end{split}
	\label{eq:sdot_solved}
	\end{align}
Note that  $\dot{S}(x,\bar x)\le0$ and hence $\bar x$ is stable.
 
\bigskip

\noindent To demonstrate our claim about the convergence of the closed-loop system's trajectories, we base the following analysis on LaSalle's invariance principle \cite[Ch.4]{khalil2} (see also Example 4.10 therein). First, we introduce the following  sublevel set of the Lyapunov function $S$:
\begin{equation}\label{eq:set_calR_kappa}
    \mathcal{R}_\kappa=\{x:~S(x,\bar{x})\leq \kappa\},
\end{equation}
We will show that  $\lim _{t\rightarrow\infty} x=\bar x'$ for a sufficiently small $\kappa$ and for some equilibrium $\bar x'\in \mathcal{E}_{\mathrm{eq}}\cap \mathcal{R}_\kappa$.

Moving forward, we introduce the following set:
\begin{align}\label{eq:set_dotS_zero}
    \mathcal{Z} & = \{x:~\dot{S}(x,\bar x)=0\},
\end{align}
 which in view of \eqref{eq:sdot_solved} is characterized by
 \begin{equation}\label{eq:set_dotS_zero_charac}
     \mathcal{Z} =  \{x:~ z_\tau=\boldsymbol{0}_{n_\mathrm{s}} \land V_\mathrm{dc}=V_\mathrm{dc}^*\}.
 \end{equation}
Due to the positive invariance of $\mathcal{R}_\kappa$, it is  possible to choose $\kappa>0$ small enough so that the following two conditions are met simultaneously: 
\begin{subequations}\label{eq:projections}
    \begin{align}
        \text{Proj}_{\varphi_\tau}(\mathcal{R}_\kappa) & \subset \mathbb{R}^{n_\mathrm{s}}_{>0}, \label{eq:proj_R_varphi}\\
        \text{Proj}_{\eta_{\tau}}(\mathcal{R}_\kappa) & \subset [\eta_{\tau,1}^\text{min},\eta_{\tau,1}^\text{max}]\times [\eta_{\tau,2}^\text{min},\eta_{\tau,2}^\text{max}]\times \nonumber\\
        & \phantom{=} \cdots \times [\eta_{\tau,n_\mathrm{s}}^\text{min},\eta_{\tau,n_\mathrm{s}}^\text{max}]\subset \mathbb{R}^{n_\mathrm{s}}. \label{eq:proj_R_hatl}
    \end{align}
\end{subequations}
We note that \eqref{eq:proj_R_varphi} is feasible due to the fact that  \( \bar{\varphi}_{\tau,i} > 0 \) for all \( i \in \mathcal{N}_\mathrm{s} \) (see \eqref{eq:desired_equilibrium}). Also, due to \eqref{eq:Tphi_conditions} and Assumption~\ref{Assumption 1}.(iv), there always exist constants $\eta_{\tau,i}^{\mathrm{max}}\ge L_{\tau,i}^{\mathrm{max}}$ and $\eta_{\tau,i}^{\mathrm{min}} \le L_{\tau,i}^{\mathrm{min}}$ with $T_{\varphi_{\tau,i}} > \eta_{\tau,i}^{\mathrm{max}} - \eta_{\tau,i}^{\mathrm{min}},$ such that \eqref{eq:proj_R_hatl} is feasible.
We continue then under the assumption that $\kappa>0$ is such that \eqref{eq:projections} holds.
Let $[t_0,\infty)\mapsto x(t)$ be an arbitrary solution of \eqref{eq:shifted_closed-loop} restricted to the set
\begin{equation}\label{eq:invariant_set_cap_dots}
   \mathcal{A}:= \mathcal{Z}\cap \mathcal{R}_\kappa\cap \mathcal{I},
\end{equation}
where we have included $\mathcal{I}$, defined in \eqref{eq:invariant_set_theta}, to recall that all solutions of \eqref{eq:shifted_closed-loop} evolve in this invariant set (see Remark~\ref{rem:invariant_set_theta}). Then, $z_\tau=\bar{z}_\tau=\boldsymbol{0}_{n_\mathrm{s}}$ and $V_\mathrm{dc}=\bar{V}_\mathrm{dc}=V_\mathrm{dc}^*$. Consequently the closed-loop dynamics \eqref{eq:closed-loop} is as follows, 
\begin{subequations}
	\begin{align}
		%
			%
			0 &= \text{diag} (\varphi_{\tau})(\hat{r}_{\tau} - r_{\tau})  - W \mathcal{L}^\text{com}\theta_{\tau}  \notag \\ 
            & \; \; \; \;- \text{diag} \left( T^{-1}_{\varphi_\tau} ( W  \mathcal{L}^\text{com} \theta_{\tau}) \right) (\eta_{\tau} - l_{\tau}), \label{subeq:LassaleEq1} \\ 
			0 &= \boldsymbol{1}_{n_\mathrm{s}}^\top \varphi_{\tau} - I_{\ell} - Y\bar{V}_{\mathrm{dc}}, \label{subeq:LassaleEq2}\\  
			T_{\varphi_\tau} \dot{\varphi}_{\tau} &= - W \mathcal{L}^\text{com}\theta_{\tau}, \label{subeq:LassaleEq3}\\
			T_{\theta_\tau} \dot{\theta}_{\tau} &= \mathcal{L}^\text{com} W\varphi_{\tau}, \label{subeq:LassaleEq4}\\
			T_{\hat{r}_\tau} \dot{\hat{r}}_{\tau} &= 0,  \label{subeq:LassaleEq5} \\
            T_{\eta_{\tau}} \dot{\eta}_{\tau} &= 0. \label{subeq:LassaleEq6}
		%
	\end{align}
\end{subequations}
From  \eqref{subeq:LassaleEq5}  and \eqref{subeq:LassaleEq6} it follows directly that $\hat{r}_\tau=\bar{\hat{r}}_\tau'$ and $\eta_{\tau} =\bar{\eta}'_\tau$ for all $t\geq t_0$, for constant $\bar{\hat{r}}_\tau'$ and $\bar{\eta}'_\tau$, respectively.
We show next that $\bar{\hat{r}}_\tau'=r_\tau$. By combining  \eqref{subeq:LassaleEq1} and \eqref{subeq:LassaleEq3} we obtain that
\begin{align}
	0 &= \text{diag} (\varphi_{\tau}) (\bar{\hat{r}}_{\tau}'-r_\tau) + T_{\varphi_\tau} \dot{\varphi}_{\tau}   + \text{diag} (\dot{\varphi}_{\tau}) (\bar{\eta}_{\tau}'-l_\tau).\label{subeq:LassaleEq1Edit}    
\end{align}
Consider the identity \( \text{diag} (\dot{\varphi}_{\tau}) (\bar{\eta}_{\tau}'-l_\tau) = \text{diag}(\bar{\eta}_{\tau}'-l_\tau) \dot{\varphi}_{\tau} \) and let us introduce the constant diagonal matrix
\begin{equation} \label{eq:DefineTl}
     T^l_{\varphi_\tau} \coloneqq T_{\varphi_\tau} + \text{diag} (\bar{\eta}_{\tau}'-l_\tau).
\end{equation}
We can rewrite \eqref{eq:DefineTl} element-wise as:
\begin{align}
     T^l_{\varphi_{\tau,i}} \coloneqq T_{\varphi_{\tau,i}} +  \bar{\eta}_{\tau,i}'-L_{\tau,i}, \quad \forall i\in \mathcal{N}_\mathrm{s},
\end{align}
Due to Assumption~\ref{Assumption 1}.(iv), the choice \( T_{\varphi_{\tau,i}} > \eta_{\tau,i}^{\mathrm{max}} - \eta_{\tau,i}^{\mathrm{min}} \), and since we analyze the closed-loop behavior within the set \( \mathcal{A} \) (see \eqref{eq:invariant_set_cap_dots}), where \( \bar{\eta}'_{\tau,i} \in [\eta_{\tau,i}^{\mathrm{min}}, \eta_{\tau,i}^{\mathrm{max}}] \), it follows that \( T^l_{\varphi_{\tau,i}} \) remains strictly positive for all \( i \in \mathcal{N}_\mathrm{s} \). Consequently, the matrix \( T^l_{\varphi_\tau} \) is positive definite. Now using \eqref{eq:DefineTl} we can reformulate \eqref{subeq:LassaleEq1Edit} as:
\begin{subequations}
\begin{align}
%
    T^l_{\varphi_\tau} \dot{\varphi}_{\tau} & =-\text{diag}(\bar{\hat{r}}_\tau'-r_\tau) \varphi_{\tau},  \label{eq:phidotnew}\\
    & \Leftrightarrow \nonumber\\
    T^l_{\varphi_{\tau,i}}{\dot{\varphi}_{\tau,i}} &= - (\bar{\hat{r}}_{\tau,i}'-r_{\tau,i})\varphi_{\tau,i},\quad \forall i\in \mathcal{N}_\mathrm{s}, \label{subeq:phidotIsCphi} \\
&  \Leftrightarrow \nonumber \\
    {\varphi}_{\tau,i} &= {\varphi}_{\tau,i}(t_0)e^{-\frac{\bar{\hat{r}}_{\tau,i}'-r_{\tau,i}}{ T^l_{\varphi_{\tau,i}}}t},\quad \forall i\in \mathcal{N}_\mathrm{s}, \label{subeq:phiSoln}
%
\end{align}
\end{subequations}
where ${\varphi}_{\tau,i}(t_0)$ denotes an initial value. Since  $x(t)\in \mathcal{A}$ for all $t\geq t_0$ (see \eqref{eq:invariant_set_cap_dots}), it holds that ${\varphi}_{\tau,i}(t_0)>0$ for all $i\in \mathcal{N}_\mathrm{s}$. In view of \eqref{subeq:phiSoln} and since $\mathcal{R}_\kappa$ is compact, then $\bar{\hat{{r}}}_{\tau,i}'-r_{\tau,i}\geq 0$ for all $i\in \mathcal{N}_\mathrm{s}$, or otherwise  $x(t)$ would be unbounded, entailing a contradiction.
Furthermore, if we consider \eqref{subeq:LassaleEq2} and use  $\varphi_{\tau,i}$ from \eqref{subeq:phiSoln}, it is possible to write that
\begin{align}
    \begin{split}
   \sum_{i=1}^{n_\mathrm{s}} {\varphi}_{\tau,i}(t_0)e^{-\frac{\bar{\hat{{r}}}_{\tau,i}'-r_{\tau,i}}{ T^l_{\varphi_{\tau,i}}}t} = I_{\ell} + Y \bar{V}_{\mathrm{dc}}. \label{eq:negatingR>0}
    \end{split}    
\end{align}
However, this equation would be rendered invalid if $\bar{\hat{{r}}}_{\tau,i}'-r_{\tau,i}>0$ for some $i\in \mathcal{N}_\mathrm{s}$. Therefore, $\bar{\hat{{r}}}_{\tau,i}'=r_{\tau,i}$.  Having established this, we now proceed to show that for all \( t \geq t_0 \), the trajectories of \( \theta_\tau \) and \( \varphi_\tau \) coincide with their corresponding equilibrium values in the set \( \mathcal{E}_{\mathrm{eq}} \).

Since $\hat{r}_\tau=r_\tau$ for all $t\geq t_0$, as established,  it follows from \eqref{eq:phidotnew} that $\dot{\varphi}_{\tau} = 0$, which in turn implies, from \eqref{subeq:LassaleEq3}, that
\begin{equation}\label{eq:stab_ana_theta_minus_bartheta}
    \theta_\tau=s(t)\boldsymbol{1}_{n_\mathrm{s}},
\end{equation}

for some $s(t)\in \mathbb{R}$.  Given that in particular $x(t)\in \mathcal{I}$ for all $t\geq t_0$ (see \eqref{eq:invariant_set_cap_dots} and recall $\mathcal{I}$ from \eqref{eq:invariant_set_theta}), then
\begin{align*}
    \boldsymbol{1}_{n_\mathrm{s}}^\top T_{\theta_\tau} \theta_\tau & = \boldsymbol{1}_{n_\mathrm{s}}^\top T_{\theta_\tau} \theta_\tau(0)\\
    & \Rightarrow \\
    \boldsymbol{1}_{n_\mathrm{s}}^\top T_{\theta_\tau} \left(s(t)\boldsymbol{1}_{n_\mathrm{s}} \right) & = \boldsymbol{1}_{n_\mathrm{s}}^\top T_{\theta_\tau} \theta_\tau(0)\\
    & \Rightarrow\\
    s(t) & = \frac{\boldsymbol{1}_{n_\mathrm{s}}^\top T_{\theta_\tau} \theta_\tau(0)}{\boldsymbol{1}_{n_\mathrm{s}}^\top T_{\theta_\tau} \boldsymbol{1}_{n_\mathrm{s}}} \; \; \forall t\geq t_0,
\end{align*}
where we have used the fact that $ \boldsymbol{1}_{n_\mathrm{s}}^\top T_{\theta_\tau}\bar{\theta}_\tau = \boldsymbol{1}_{n_\mathrm{s}}^\top T_{\theta_\tau} \theta_\tau(0)$ (due to $x(t)\in \mathcal{I}$). Hence, $\theta_\tau =\beta\boldsymbol{1}_{n_\mathrm{s}}$ for all $t\geq t_0$. 

Having established that $\theta_\tau= \beta \boldsymbol{1}_{n_\mathrm{s}}$ for all $t\geq t_0$ (which implies that $\dot{\theta}_\tau=0$), then  from \eqref{subeq:LassaleEq4} it follows that
\begin{equation*}
    \varphi_\tau = \hat{s}(t)W^{-1}\boldsymbol{1}_{n_\mathrm{s}},
\end{equation*}
for some $\hat{s}(t)\in \mathbb{R}$. From \eqref{subeq:LassaleEq2} we get 
 that
\begin{align*}
     \boldsymbol{1}_{n_\mathrm{s}}^\top\varphi_\tau=\boldsymbol{1}_{n_\mathrm{s}}^\top \left( \hat{s}(t)W^{-1}\boldsymbol{1}_{n_\mathrm{s}}\right) &=  I_{\ell} + Y\bar{V}_{\mathrm{dc}}\\
    & \Rightarrow\\
     \hat{s}(t) & = \frac{I_{\ell} + Y\bar{V}_{\mathrm{dc}}}{\boldsymbol{1}_{n_\mathrm{s}}^\top  W^{-1}\boldsymbol{1}_{n_\mathrm{s}}} \; \; \forall t\geq t_0.
\end{align*}
Consequently, $\varphi_\tau = \alpha W^{-1}\boldsymbol{1}_{n_\mathrm{s}}$ for all $t\geq t_0$.  
The above analysis shows that the largest invariant set of the closed-loop dynamics, restricted to $\mathcal{R}_\kappa\cap \mathcal{I}$ and for a sufficiently small $\kappa>0$, is given by some equilibrium of the form
\begin{equation}
    \bar x'=(\boldsymbol{0}_{n_\mathrm{s}},C_\mathrm{dc}V_\mathrm{dc}^*,\alpha T_{\varphi_\tau} W^{-1}\boldsymbol{1}_{n_\mathrm{s}},\beta T_{\theta_\tau} \boldsymbol{1}_{n_\mathrm{s}},T_{\hat{r}_\tau}r_\tau, T_{\eta_{\tau}}\bar{\eta}_{\tau}'),
\end{equation}
where $\alpha$ and $\beta$ are as in \eqref{eq:values_alpha_beta}, and $\bar{\eta}'$ is some constant vector. Such an equilibrium  is in agreement with Proposition~\ref{prop:eqAchieveObj}. Consequently, by virtue of LaSalle's invariance principle \cite[Ch.4]{khalil2}, all closed-loop trajectories starting sufficiently close to $\bar x$ in \eqref{eq:desired_equilibrium} will converge to an equilibrium $\bar x' \in \mathcal{E}_{\mathrm{eq}}$, implying the satisfaction of Objectives~\ref{obj:VoltReg} and \ref{obj:CurrShar}. 
$\hfill\blacksquare$

\end{proof}

\vspace{0.2cm}

\section{Experimental Results} \label{Sec:Experiments}

{\color{black}
In this section, we experimentally evaluate the performance of the proposed controller at a \gls{phil} laboratory setup emulating the \gls{epds} level of a hybrid-electric propulsion system. The objective is to test the controller capabilities under realistic electrical disturbances (e.g., measurement noise, communication delays, and parameter uncertainties), thereby demonstrating the robustness of the control strategy. This subsystem-level validation provides a foundation for future investigations under more comprehensive aero-thermal conditions. The proposed controller is benchmarked against the state-of-the-art droop control method commonly employed in \gls{mg} control, as well as against the controller presented in \cite{Trip2}.

It is important to note that the controller in \cite{Trip2} was developed for a different setup. In \cite{Trip2}, multiple \gls{dgu}s are connected to corresponding PCCs, whereas in our work, multiple electric sources are connected to a single non-actuated load bus representing the main DC link of the propulsion system. Moreover, \cite{Trip2} targets average voltage regulation among generation buses rather than voltage regulation at the load bus. The comparison is included because the controller in \cite{Trip2} has a similar distributed structure and can be implemented in the considered \gls{epds} setting when line parameters are known and \gls{dgu} inductor currents are interpreted as line currents. Under these assumptions, the controller in \cite{Trip2} can be tuned to achieve voltage regulation in the considered setup. It therefore serves as a representative parameter-dependent benchmark against which the proposed adaptive controller can be compared.
}

\subsection{Experimental Setup} \label{subsec:ExperimentalSetup}
\newcommand{\ifilt}[1]{i^\mathrm{f}_{#1}}
\newcommand{\iout}[1]{I_{#1}}
\newcommand{\uout}[1]{u_{#1}}
\newcommand{\ubus}{V_\mathrm{dc}}

To experimentally validate the controller design, we performed a series of tests in the \gls{phil} Laboratory at BTU Cottbus-Senftenberg~\cite{krishna_power-hardware---loop_2022}. Next we introduce some local notation to help us represent the dynamics of parts of the system setup that were neglected or simplified earlier.

We have used four converters with a power rating of \SI{15}{kW} each, and each of which is programmable through an independent \gls{rtc}. Converters 1-3 are equipped with an LC output filter, while Converter-4 is equipped with an LCL output filter. The experimental setup is shown in Fig.~\ref{fig:Expsetup}  and a schematic overview is provided in Fig.\ref{fig:labsetup}.

Converters 1-3 are operated in voltage control mode, with the output voltage of the $i$-th converter  given by the voltage $u_i$ across the capacitor of the LC filter. The voltage setpoint $\uout{i}^*$  is provided by the controller \eqref{eq:controllerInDistributed}, where we find it convenient to  recall that an underlying assumption behind model~\eqref{eq:DcMgHepDiffEqn} is that \gls{dgu}s' internal dynamics are negligible. In addition, Converter-4  is operated in current-control mode to represent the electrical load.

The three voltage source converters are connected via power lines to the main DC bus (i.e., \gls{pcc}). Each power line consists of two capacitors located at both ends, with an inductor and a resistor positioned between them. The capacitor of the power line connected to the voltage source converter can be considered combined with the capacitor of the converter's LC filter, effectively treating them as a single lumped component. Additionally, the capacitors of the power lines connected to the PCC can be combined into a single equivalent capacitor with capacitance \( C_\mathrm{dc} \). The current-controlled converter is directly connected to the \gls{pcc}.
In this setup, the voltage-controlled converters mimic the DGUs in a series-hybrid-electric propulsion system, while the current-controlled converter mimics the behavior of the \gls{epu}. This approach ensures that the experimental setup closely mimics the model \eqref{eq:DcMgHepDiffEqn}. 

Each converter can access local measurements of filter currents $\ifilt{i}$, $\iout{\tau,i}$, and filter voltage $\uout{i}$. Then, we have set up a communication network $\mathcal{G}^{\mathrm{com}}$: the current-controlled Converter-4 measures the main bus voltage $\ubus$ and broadcasts this measurement to converters 1-3; Converters 1-3 share their respective output current measurement $\iout{\tau}$ and control variable $\theta_{\tau}$ with their neighbor(s) in the communication network $\mathcal{G}^{\mathrm{com}}$. We note that each communication link introduces a delay of approximately \SI{2}{ms}.

\begin{figure*}
  \centering
  \includegraphics[width=1\textwidth]{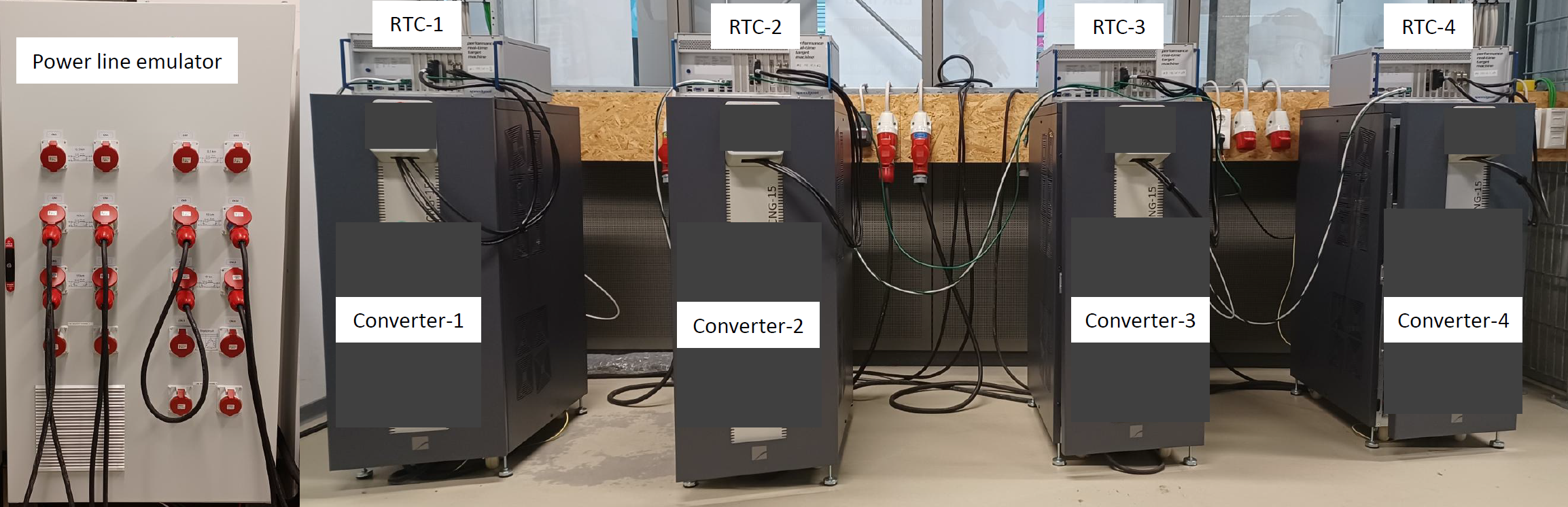}
  \caption{Experimental setup for the \gls{phil} testing.}\label{fig:Expsetup}
\end{figure*}

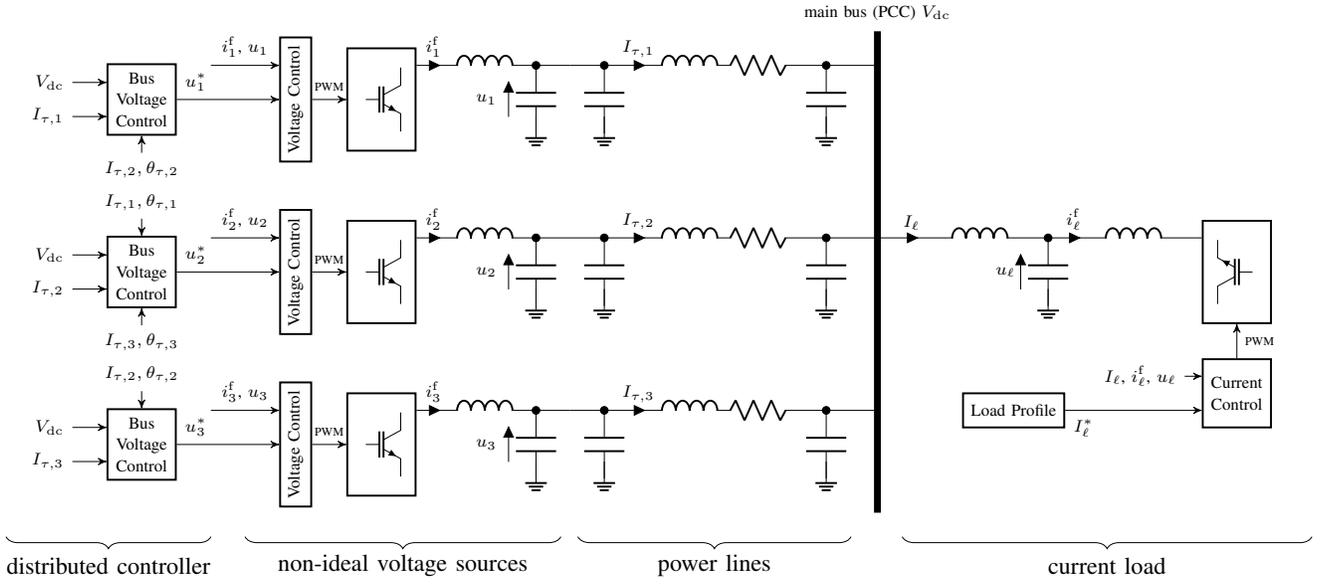
\begin{figure*}
  \centering
  \resizebox{\textwidth}{!}{
  \newcommand{\voltagesource}[2]{
  \node[block, minimum width=1cm, minimum height=1.5cm, #2](igbt#1){};
  \node at (igbt#1) [nigbt, scale=0.6, anchor=centergap]{} ;
  \draw (igbt#1.east 1)
  to [short,i=\scriptsize $\ifilt{#1}$] ++(.5,0)
  to [L] ++(1,0)
  to [short,-*] ++(.25,0) coordinate(csplit){}
  to [short] ++(0.5,0) coordinate (terminal#1);
  \draw (csplit) to [C, v<=\scriptsize $\uout{#1}\,\,$] ++(0,-1);
  \path (csplit) ++(0,-0.75) node[ground, scale=.75]{};
  \node[block,left=0.5cm of igbt#1, minimum width=0cm](vc#1){\rotatebox{90}{\scriptsize Voltage Control}};
  \draw[connection](vc#1)--(igbt#1) node[pos=.5,above]{\tiny PWM};
  \draw[connection,<-](vc#1.west 1) --++ (-1,0) node[pos=.5, above]{\scriptsize $\ifilt{#1}$, $\uout{#1}$};
}
\newcommand{\piline}[1]{
  \draw (terminal#1) to [short] ++(.75,0)
  to [short,i=\scriptsize $\iout{\tau,#1}$] ++(.5,0)
  to [L]++(1,0)
  to [R]++(1,0)
  to[short]++(1.25,0) coordinate (lineterminal#1);
  \draw (terminal#1)++(0.5,0) to [C,*-] ++(0,-1);
  \path (terminal#1)++(0.5,0) ++(0,-0.75) node[ground, scale=.75]{};
  \draw (lineterminal#1)++(-0.75,0) to [C,*-] ++(0,-1);
  \path (lineterminal#1)++(-0.75,0) ++(0,-0.75) node[ground, scale=.75]{};
}
\newcommand{\buscontrol}[1]{
  \node[block, left=1.5cm of vc#1](buscontrol#1){\scriptsize Bus\\\scriptsize Voltage\\\scriptsize Control};
  \draw[connection](buscontrol#1.east) node[above right]{\scriptsize $\uout{#1}^{*}$}--(vc#1);
  \draw[connection, <-](buscontrol#1.west A)--++(-.5,0) node[left]{\scriptsize $\ubus$};
  \draw[connection, <-](buscontrol#1.west B)--++(-.5,0) node[left]{\scriptsize $\iout{\tau,#1}$};
}
\begin{tikzpicture}
  \voltagesource{1}{};
  \voltagesource{2}{below=1cm of igbt1};
  \voltagesource{3}{below=1cm of igbt2};

  \piline{1};
  \piline{2};
  \piline{3};

  \draw[line width=0.1cm] (lineterminal1)++(0,0.5) node[above]{\scriptsize main bus (PCC) $\ubus$} --(lineterminal3)--++(0,-1.5);

  \draw (lineterminal2)   
  to [short,i=\scriptsize $\iout{\ell}$] ++(1,0)
  to [L] ++(1,0)
  to [short,-*] ++(.5,0) coordinate(csplit){}
  to [short,i=\scriptsize $\ifilt{\mathrm{\ell}}$] ++(.75,0)
  to [L] ++(1,0)
  to [short] ++(0.5,0) coordinate (lineend);
  \draw (csplit) to [C, v<=\scriptsize $\uout{\mathrm{\ell}}$] ++(0,-1);
  \path (csplit) ++(0,-0.75) node[ground, scale=.75]{};

  \node[right=0 cm of lineend, block, minimum width=1cm, minimum height=1.5cm, yshift=-.5cm](igbtB){};
  \node at (igbtB) [nigbt, scale=0.6, anchor=centergap, rotate=180]{} ;
  \node[block, below=.5cm of igbtB, minimum height=1cm](CC){\scriptsize Current\\\scriptsize Control};
  \draw[connection](CC)--(igbtB) node[pos=.5, right]{\tiny PWM};
  \node[block, left=2cm of CC, yshift=-.25cm](loadprofile){\scriptsize Load Profile};
  \draw[connection](loadprofile.east) node[below right]{\scriptsize $\iout{\mathrm{\ell}}^{*}$}--(CC.west B);
  \draw[connection,<-](CC.west A) --++ (-.25,0) node[left]{\scriptsize $\iout{\mathrm{\ell}}$, $\ifilt{\mathrm{\ell}}$, $\uout{\mathrm{\ell}}$};

  \buscontrol{1}
  \draw[connection,<-](buscontrol1.south)--++(0,-.25)node[below]{\scriptsize $\iout{\tau,2}, \theta_{\tau,2}$};

  \buscontrol{2}
  \draw[connection,<-](buscontrol2.north)--++(0, .25)node[above]{\scriptsize $\iout{\tau,1}, \theta_{\tau,1}$};
  \draw[connection,<-](buscontrol2.south)--++(0,-.25)node[below]{\scriptsize $\iout{\tau,3}, \theta_{\tau,3}$};

  \buscontrol{3}
  \draw[connection,<-](buscontrol3.north)--++(0, .25)node[above]{\scriptsize $\iout{\tau,2}, \theta_{\tau,2}$};

  \path (buscontrol3.north east)++(1,-1.25) coordinate (curly0end);
  \draw [decorate,decoration={brace,amplitude=5pt,mirror,raise=4ex}]
  (buscontrol3.north east)++(1,-1.25) --++ (4.625,0) coordinate (curly1end) node[midway,yshift=-3em]{non-ideal voltage sources};
  \draw [decorate,decoration={brace,amplitude=5pt,mirror,raise=4ex}]
  (curly1end)++(0.25,0) --++ (4.0,0) coordinate (curly2end) node[midway,yshift=-3em](curly1end){power lines};
  \draw [decorate,decoration={brace,amplitude=5pt,mirror,raise=4ex}]
  (curly2end)++(0.75,0) --++ (6,0) node[midway,yshift=-3em](curly1end){current load};

  \draw [decorate,decoration={brace,amplitude=5pt,mirror,raise=4ex}]
  (buscontrol3.north east)++(-2.5,-1.25) --++ (3,0) node[midway,yshift=-3em]{distributed controller};
\end{tikzpicture}
  }
  \caption{Schematic overview of the experimental setup for the \gls{phil} testing. }\label{fig:labsetup}
\end{figure*}



\subsection{Results}

We assess the performance of the proposed control scheme using the \gls{phil} setup described in Section \ref{subsec:ExperimentalSetup}. Its performance is quantitatively compared to that of the droop controller and the controller presented in \cite{Trip2}. For the sake of brevity, we refer to the proposed controller as C1, the droop controller as C2, and the controller from \cite{Trip2} as C3. 

The system parameters are given in Table~\ref{table:simulationparameters}. The controller  gains for C1 are $T_{\varphi_\tau} = \mathbb{I}_3$, $T_{\theta_\tau} = \mathbb{I}_3$, $T_{\hat{r}_\tau} = 10\mathbb{I}_3$, $T_{\eta_{\tau}} = 10^6\mathbb{I}_3$ and $K_{z_\tau} = 2\mathbb{I}_3$. Without loss of generality, we take  $W = \mathbb{I}_3$, ensuring equal current from the three \gls{dgu}s at a steady-state to clearly show the accomplishment of Objective~\ref{obj:CurrShar}. The controller gains are also chosen for C2 and C3 using trial and error for a suitable transient response. C3 requires the power line resistances $R_\tau$ to be known. To account for realistic conditions, the resistances are set to their actual values but with a 10\% error, reflecting the fact that, in real scenarios, power line resistances may not be exactly known and have to be estimated. 

In Table~\ref{table:flightSegments}, we display the scaled propulsion load current demand ($I_\ell$) associated with selective flight segments within a regional mission profile for a series-hybrid-electric aircraft. Three flight segments are considered to demonstrate that the controller can achieve its objectives at different power levels. Each flight segment time has been scaled for better visualization, but notably, the controller's response remains unscaled. Furthermore, due to the limited power-handling capacity of our experimental setup, we scaled down the power requirements for each flight segment accordingly. {\color{black} Since the controller design and stability analysis do not rely on the numerical value of $I_\ell$, but only on its positivity and boundedness, this scaling does not affect the qualitative stability and performance conclusions.} The base voltage of the main DC bus and base propulsion load current are $V^{\mathrm{base}}_\mathrm{dc} = 200 \si{\volt}$ and $I^{\mathrm{base}}_l = 6.7 \si{\ampere}$, respectively.

\begin{center}
 {\footnotesize
	\captionof{table}{System Parameters.} \label{table:simulationparameters}
	\begin{tabular}{ |l |c |}
		 \hline
		Model Parameters & Values \\
		 \hline
		  Line inductances, $L_{\tau,i}$, $i = 1,2,3.$ & [900, 550, 350] \si{\micro\henry} \\
		\hline  
		Line resistances, $R_{\tau,i}$, $i = 1,2,3.$ & [1.33, 0.78, 0.71] $\Omega$\\
            \hline 
            Main bus capacitance $C_\mathrm{dc}$ & 0.318 \si{\micro\farad}\\
            \hline

	\end{tabular}
 }
\end{center}


\begin{center}
 {\footnotesize
	\captionof{table}{Scaled Regional Mission Profile for HEA.} \label{table:flightSegments}
	\begin{tabular}{ |l |c |c|} 
		 \hline
		Segment & Duration (sec) & $I_\ell$ (p.u.) \\
		 \hline
		  Takeoff & 35 & 2.98 \\
		\hline  
		Cruise  &  25 &  2.3 \\
            \hline 
            Landing  &  25 &  1.7 \\
            \hline

	\end{tabular}
 }
\end{center}

The results of our experimental test for C1 are displayed in Fig.~\ref{fig:Experiments}. The evolution of the load bus voltage $V_\mathrm{dc}$, the DGUs' injected currents $I_\tau$, the variable $\hat{r}_{\tau}$ used to estimate the line resistance matrix $R_\tau$  are illustrated in Fig.\ref{subfig:voltageRegulation}, Fig.\ref{subfig:currentSharing}, Fig.\ref{subfig:resistanceEstimation}, respectively. It can be observed in Fig.\ref{subfig:voltageRegulation} and Fig.\ref{subfig:currentSharing} that despite the changes in the load current demand, the control Objectives~1 and 2 are satisfied. The controller C1 also estimates the power line resistances with reasonable accuracy, as shown in Fig.\ref{subfig:resistanceEstimation}.

The performance of the proposed controller, C1, was compared  against the controllers C2 and C3 from the literature. In particular, we compared the capacity to achieve voltage regulation and load-sharing.  The results are respectively shown in Fig.~\ref{fig:Comparison}. In particular ~\ref{subfig:voltageDeviation} shows the absolute  voltage deviation (in \%) from the reference voltage $V_\mathrm{dc}^*$ for the  three controllers. Note that C1 exhibits minimal deviation from the reference voltage, while C2 displays the largest deviation. 
Fig. \ref{subfig:currentDeviation} illustrates the deviation from the expected consensus in load-sharing for the three controllers. To be precise, the Euclidean norm of the sum of current differences ($I_{\tau,i}-I_{\tau,j} \;  \; \forall \;  \; i,j \; \epsilon \;  \; \mathcal{N}_\mathrm{s}$) is plotted over time. Note that C1 and C3 demonstrate good performance in achieving Objective~2, exhibiting comparable results. In contrast, C2 shows the largest deviation from the expected consensus in load-sharing

In summary, C1 demonstrates superior overall performance compared to C2 and C3. The controller C2 neglects power line resistance entirely, while C3 relies on precise prior knowledge of it. In contrast, C1 adaptively estimates power line resistance in real-time, which is particularly advantageous in practical systems where full system parameters are often unknown. This adaptability allows C1 to better accommodate system uncertainties, resulting in improved overall performance.

\begin{figure}[hbt!]
\centering

\subfloat[Voltage across the PCC. \label{subfig:voltageRegulation}]
{\includegraphics[width=0.95\columnwidth]{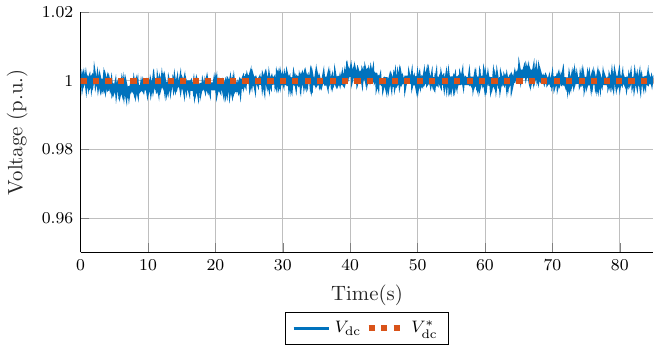}}

\vspace{1em}

\subfloat[Currents generated by the \gls{dgu}s. \label{subfig:currentSharing}]
{\includegraphics[width=0.95\columnwidth]{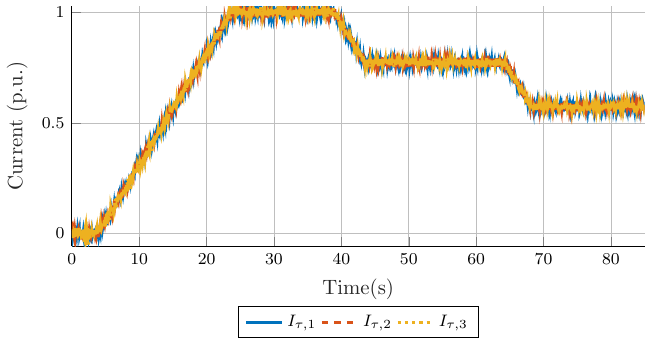}}

\vspace{1em}

\subfloat[Comparison of power line resistance estimations to measured values. \label{subfig:resistanceEstimation}]
{\includegraphics[width=0.95\columnwidth]{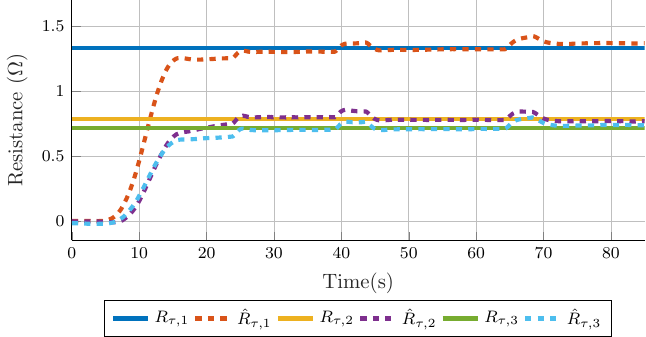}}

\caption{Experimental validation of the proposed controller.}
\label{fig:Experiments}
\end{figure}

\begin{figure}[hbt!]
\centering
\subfloat[Comparison of voltage regulation performance. \label{subfig:voltageDeviation}]
{\includegraphics[width=0.95\columnwidth]{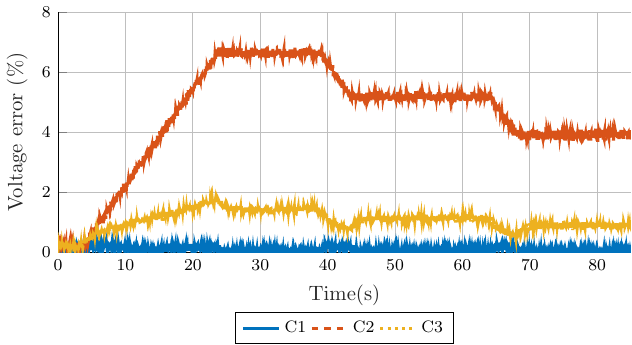}}\\
\subfloat[Comparison of current sharing performance. \label{subfig:currentDeviation}]
{\includegraphics[width=0.95\columnwidth]{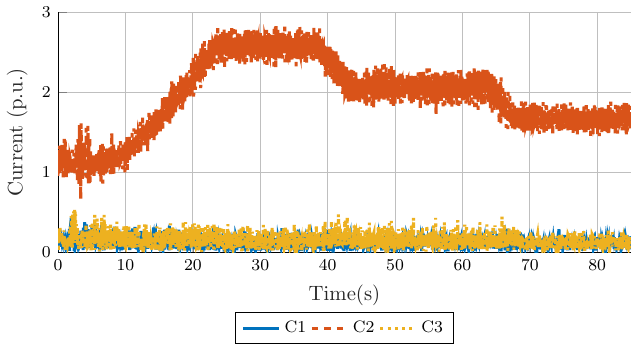}}\\
\caption{Performance comparison for the controllers C1, C2 and C3.}
\label{fig:Comparison}
\end{figure}

\section{CONCLUSIONS} \label{Sec:Conclusion}

In this paper, a distributed control algorithm has been proposed to achieve voltage regulation and load-sharing for an \gls{epds} of a series-hybrid-electric propulsion system.
The considered on-board \gls{epds} model includes non-negligible, unknown power line resistances and unknown but bounded inductances. These factors hinder the application of existing control algorithms developed for similar objectives in DC \gls{mg}s.
Our controller suitably combines the principles of back-stepping, adaptive control, and passivity-based control methods. 
Notably, the adaptive feature of the controller allows us to mitigate the challenges brought by the presence of non-negligible power line resistances and inductances, and eases the stability and convergence analysis of the resulting closed-loop system.
Additionally, we evaluated the performance of the proposed control scheme through  an experimental  \gls{phil} setup,  showing a superior performance against two state-of-the-art controllers.

{\color{black} As future research, we plan to incorporate the bilinear internal dynamics of the power converters into the considered model and investigate stability and convergence under scenarios involving connection and disconnection of \glspl{dgu}. This will require accounting for the resulting switching communication topologies, thereby enabling rigorous plug-and-play operation of electrical sources. Furthermore, building on the foundation provided by the \gls{phil} experiments, future work will include validating the controller under representative aerothermal conditions at the Center for Hybrid Electric Systems Cottbus (CHESCO) facility.}

\section*{Appendix}

In this Appendix, we provide further details behind the conception of the proposed controller \eqref{eq:controller}, which as mentioned in   Section~\ref{subsec:PropsedControl}, is based on the general design methodology  IDA-PBC with dynamic extension \cite{astolfi_ortega}. In our particular case, the design of the vector fields defining the dynamic extension is based on proportional consensus dynamics and back-stepping.

Consider the following general non-linear control system:
\begin{equation}\label{eq:PhModel}
\dot x_\mathrm{p} = f(x_\mathrm{p})+g(x_\mathrm{p})u,
\end{equation}
where $x_p \in \mathbb{R}^n$, $f(x)\in \mathbb{R}^n$, $g(x)\in \mathbb{R}^{n\times m}$ and $u\in \mathbb{R}^m$. In a nutshell, the design idea of IDA-PBC with dynamic extension is to develop a dynamic feedback law of the form
\begin{align}
    \begin{split}
        u = \gamma(x_\mathrm{p},\xi), \\
        \dot{\xi} = \rho(x_\mathrm{p},\xi),
    \end{split}
\end{align}
where $\xi\in \mathbb{R}^s$ are the controller's states, and $\gamma$ and $\rho$ are continuous functions, such that the closed-loop system has a port-Hamiltonian structure \cite{arjan_book_port_Ham}:
\begin{align}
\begin{split}
\text{col}(\dot{x}_\mathrm{p}, \dot{\xi})   &= (\mathcal{J}_\mathrm{d}(x_\mathrm{p},\xi) - \mathcal{R}_\mathrm{d}(x_\mathrm{p},\xi))\nabla H_\mathrm{d}(x_\mathrm{p},\xi),
\end{split}
\label{eq:PhClosed-loop}
\end{align}
where $\mathcal{J}_\mathrm{d}$ is skew-symmetric, $\mathcal{R}_\mathrm{d}$ is symmetric positive semi-definite, and $H_\mathrm{d}: \mathbb{R}^n \times \mathbb{R}^s \rightarrow \mathbb{R}$ represents a Hamiltonian function. Here, $\nabla H_\mathrm{d}$ denotes the transposed gradient of $H_\mathrm{d}$ with respect to $\text{col}(x_\mathrm{p},\xi)$. A sufficient condition for \eqref{eq:PhClosed-loop} to admit a stable equilibrium is that $H_\mathrm{d}$ has a local minimum at it.

Consider then the system model \eqref{eq:DcMgHepDiffEqn} in open-loop, along with the dynamic extension proposed in \eqref{eq:controller}, and propagate the change of variables \eqref{eq:NewVariable}. This results in the following dynamics: 
\begin{align}
	\begin{split}		
		L_{\tau} \dot{z}_{\tau} &= -R_{\tau}(\varphi_{\tau} +z_{\tau}) - \boldsymbol{1}_{n_\mathrm{s}}V_\mathrm{dc} + u \\
		& \; \; \; -L_{\tau}T_{\varphi_\tau}^{-1}(-(V_\mathrm{dc} - V^*_\mathrm{dc})\boldsymbol{1}_{n_\mathrm{s}} - W \mathcal{L}^\text{com} \theta_{\tau}), \\
		C_\mathrm{dc} \dot{V}_\mathrm{dc} &= \boldsymbol{1}_{n_\mathrm{s}}^\top (z_{\tau} + \varphi_{\tau}) -I_{\ell} -  Y V_\mathrm{dc},\\
		T_{\varphi_\tau} \dot{\varphi}_{\tau} &= -(V_\mathrm{dc} - V^*_\mathrm{dc})\boldsymbol{1}_{n_\mathrm{s}} - W \mathcal{L}^\text{com} \theta_{\tau}, \\
		T_{\theta_\tau} \dot{\theta}_{\tau} &= \mathcal{L}^\text{com} W (z_{\tau} + \varphi_{\tau}), \\
		T_{\hat{r}_\tau} \dot{\hat{r}}_{\tau} &= -\text{diag} (\varphi_{\tau} + z_{\tau})z_{\tau},\\
        T_{\eta_{\tau}} \dot{\eta}_{\tau} &= -\text{diag} \left( T^{-1}_{\varphi_\tau} (-(V_\mathrm{dc} - V^*_\mathrm{dc})\boldsymbol{1}_{n_\mathrm{s}} - W \mathcal{L}^\text{com} \theta_{\tau}) \right)  z_{\tau},\\
	\end{split}
	\label{eq:open-loop+dynaext}
\end{align}
In order to compute the feedback law for $u$ that would render a closed-loop system as in \eqref{eq:PhClosed-loop}, we assign $\mathcal{J}_d$, $\mathcal{R}_\mathrm{d}$ and $H_\mathrm{d}$ as follows:
\begin{align*}
\mathcal{J}_\mathrm{d} & =\begin{bmatrix}
\mathcal{J}_\mathrm{d}^{(1,1)} & \mathcal{J}_\mathrm{d}^{(1,2)}\\
-(\mathcal{J}_\mathrm{d}^{(1,2)})^\top & \boldsymbol{0}_{3n_\mathrm{s}\times 3n_\mathrm{s}}
\end{bmatrix},\\
\mathcal{R}_\mathrm{d} & = \text{diag}\left(K_{z_\tau},Y,\boldsymbol{0}_{4n_\mathrm{s} \times 4n_\mathrm{s}}\right)\\
H_\mathrm{d}(x_\mathrm{p},\xi) & = \frac{1}{2} \| \text{col}(x_p,\xi) - \text{col}(\bar{x}_p, \bar{\xi}) \|^2_{\mathcal{M}_d},
\end{align*}
where
\begin{align*}
    \mathcal{J}_\mathrm{d}^{(1,1)} & = \begin{bmatrix} 0_{n_\mathrm{s} \times n_\mathrm{s}} & -\boldsymbol{1}_{n_\mathrm{s}} & 0_{n_\mathrm{s} \times n_\mathrm{s}} \\
                   \boldsymbol{1}_{n_\mathrm{s}}^\top & 0 & \boldsymbol{1}_{n_\mathrm{s}}^\top \\
                   0_{n_\mathrm{s} \times n_\mathrm{s}} & -\boldsymbol{1}_{n_\mathrm{s}} & 0_{n_\mathrm{s} \times n_\mathrm{s}} 
                   \end{bmatrix},\\
 \mathcal{J}_\mathrm{d}^{(1,2)} & = \begin{bmatrix}
     -W\mathcal{L}^\mathrm{com} & \text{diag}(\varphi_\tau+z_\tau) & \text{diag}(\mathcal{F})\\
     0_{n_\mathrm{s}}^\top & 0_{n_\mathrm{s}}^\top & 0_{n_\mathrm{s}}^\top\\
     -W\mathcal{L}^\mathrm{com} & 0_{n_\mathrm{s}\times n_\mathrm{s}} & 0_{n_\mathrm{s}\times n_\mathrm{s}} \end{bmatrix} ,\\
     \mathcal{M}_d & = \text{diag}^{-1}\left(L_{\tau}, C_{\mathrm{dc}}, T_{\varphi_\tau}, T_{\theta_\tau}, T_{\hat{r}_\tau}, T_{\eta_{\tau}}\right),
\end{align*}
with
\begin{equation*}
    \mathcal{F}=T^{-1}_{\varphi_\tau} (-(V_\mathrm{dc} - V^*_\mathrm{dc})\boldsymbol{1}_{n_\mathrm{s}} - W \mathcal{L}^\text{com} \theta_{\tau}).
\end{equation*}

Via direct computations it can be verified that the above choices of $\mathcal{J}_d$, $\mathcal{R}_\mathrm{d}$ and $H_\mathrm{d}$ satisfy the  \emph{extended} matching equation \cite[Equation~(5)]{astolfi_ortega}. Therefore,  it is possible to use the formula in \cite[Equation~(6)]{astolfi_ortega} to compute the feedback law for $u$ shown in \eqref{eq:controllerInDistributed}. Note that (local) stability of the equilibrium $\bar x$ in \eqref{eq:desired_equilibrium} follows from the fact that $\bar x$ represents a critical point of $H_\mathrm{d}$ and from the fact that $\dot H_\mathrm{d}=-(\nabla H_\mathrm{d})^\top \mathcal{R}_\mathrm{d}\nabla H_\mathrm{d}\leq 0$ along system trajectories.

\section*{Acknowledgement}
The ChatGPT tool \cite{chatgpt} was used to improve the syntax and grammar of certain paragraphs in the manuscript. The ChatGPT tool \cite{chatgpt} was also employed to assist in identifying relevant references related to the power control of hybrid-electric aircraft.


\bibliographystyle{IEEEtran}
\bibliography{References}

@ARTICLE{DoffSotta,
  author={Doff-Sotta, Martin and Cannon, Mark and Bacic, Marko},
  journal={IEEE Transactions on Control Systems Technology}, 
  title={Predictive Energy Management for Hybrid Electric Aircraft Propulsion Systems}, 
  year={2023},
  volume={31},
  number={2},
  pages={602-614},
  doi={10.1109/TCST.2022.3193295}}

@ARTICLE{Tucci,
  author={Tucci, Michele and Riverso, Stefano and Ferrari-Trecate, Giancarlo},
  journal={IEEE Transactions on Control Systems Technology}, 
  title={Line-Independent Plug-and-Play Controllers for Voltage Stabilization in DC Microgrids}, 
  year={2018},
  volume={26},
  number={3},
  pages={1115-1123},
  doi={10.1109/TCST.2017.2695167}}

@Article{CheYanbo,
AUTHOR = {Che, Yanbo and Xu, Jianmei and Shi, Kun and Liu, Huanan and Chen, Weihua and Yu, Dongmin},
TITLE = {Stability Analysis of Aircraft Power Systems Based on a Unified Large Signal Model},
JOURNAL = {Energies},
VOLUME = {10},
YEAR = {2017},
NUMBER = {11},
ARTICLE-NUMBER = {1739},
ISSN = {1996-1073},
DOI = {10.3390/en10111739}
}

@incollection{ASHRAE_Aircraft_2019,
  title     = {Aircraft},
  booktitle = {ASHRAE Handbook---HVAC Applications},
  chapter   = {13},
  year      = {2019},
  publisher = {American Society of Heating, Refrigerating and Air-Conditioning Engineers (ASHRAE)},
}

@article{Benjamin,
title = {Electric, hybrid, and turboelectric fixed-wing aircraft: A review of concepts, models, and design approaches},
journal = {Progress in Aerospace Sciences},
volume = {104},
pages = {1-19},
year = {2019},
issn = {0376-0421},
doi = {https://doi.org/10.1016/j.paerosci.2018.06.004},
author = {Benjamin J. Brelje and Joaquim R.R.A. Martins},
}

@article{Potamiti, 
title={Thermal management system design for a series hybrid-electric propulsion architecture}, 
volume={128}, 
DOI={10.1017/aer.2023.111}, 
number={1325}, 
journal={The Aeronautical Journal}, 
author={Potamiti, M. and Gkoutzamanis, V.G. and Kalfas, A.I.}, 
year={2024}, 
pages={1532–1555}}

@Article{Ouyang,
AUTHOR = {Ouyang, Zeyu and Nikolaidis, Theoklis and Jafari, Soheil and Pontika, Evangelia},
TITLE = {Integrated Power and Thermal Management System in a Parallel Hybrid-Electric Aircraft: An Exploration of Passive and Active Cooling and Temperature Control},
JOURNAL = {Engineering Proceedings},
VOLUME = {90},
YEAR = {2025},
NUMBER = {1},
ARTICLE-NUMBER = {36},
ISSN = {2673-4591},
DOI = {10.3390/engproc2025090036}
}

@ARTICLE{Guerrero,
  author={Guerrero, Josep M. and Vasquez, Juan C. and Matas, José and de Vicuna, Luis García and Castilla, Miguel},
  journal={IEEE Transactions on Industrial Electronics}, 
  title={Hierarchical Control of Droop-Controlled AC and DC Microgrids—A General Approach Toward Standardization}, 
  year={2011},
  volume={58},
  number={1},
  pages={158-172},
  doi={10.1109/TIE.2010.2066534}}

@article{BAI,
title = {Voltage regulation and current sharing for multi-bus DC microgrids: A compromised design approach},
journal = {Automatica},
volume = {142},
pages = {110340},
year = {2022},
issn = {0005-1098},
doi = {https://doi.org/10.1016/j.automatica.2022.110340},
author = {Handong Bai and Hongwei Zhang and He Cai and Johannes Schiffer},
}

@misc{chatgpt,
  author       = {{OpenAI}},
  title        = {{ChatGPT}: Language model (GPT-4)},
  year         = {2024}
}

@book{khalil2,
  title={Nonlinear systems},
  author={Khalil, Hassan K and Grizzle, Jessy W},
  volume={3},
  year={2002},
  publisher={Prentice hall Upper Saddle River, NJ}
}

@book{slotine,
  title={Applied nonlinear control},
  author={Slotine, Jean-Jacques E and Li, Weiping and others},
  volume={199},
  number={1},
  year={1991},
  publisher={Prentice hall Englewood Cliffs, NJ}
}

@ARTICLE{Morteza,
  author={Monfared, Morteza Nazari and Kawano, Yu and Cucuzzella, Michele},
  journal={IEEE Transactions on Control Systems Technology}, 
  title={Decentralized Voltage Control of Boost Converters in DC Microgrids: Feasibility Guarantees}, 
  year={2024},
  volume={},
  number={},
  pages={1-0},
  doi={10.1109/TCST.2024.3440228}}

@INPROCEEDINGS{Karunarathne,
  author={Karunarathne, Lakmal and Economou, John T. and Knowles, Kevin},
  booktitle={2008 IEEE Vehicle Power and Propulsion Conference}, 
  title={Fuzzy Logic control strategy for Fuel Cell/Battery aerospace propulsion system}, 
  year={2008},
  volume={},
  number={},
  pages={1-5},
  doi={10.1109/VPPC.2008.4677772}}

@article{Hoenicke,
title = {Power management control and delivery module for a hybrid electric aircraft using fuel cell and battery},
journal = {Energy Conversion and Management},
volume = {244},
pages = {114445},
year = {2021},
issn = {0196-8904},
author = {Pia Hoenicke and Debjani Ghosh and Adel Muhandes and Sumantra Bhattacharya and Christiane Bauer and Josef Kallo and Caroline Willich},
}

@article{bastos,
  title={Control strategy for an interleaved bidirectional DC--DC converter applied to battery management in a hybrid aircraft propulsion system},
  author={Bastos, Maria Eduarda S and Nascimento, Saulo O and Torres, Vitor CS and Almeida, Matheus S and Rend{\'o}n, Manuel A and Rodrigues, M{\'a}rcio CBP and Almeida, Pedro S and Oliveira, Jana{\'\i}na G},
  journal={Journal of Control, Automation and Electrical Systems},
  volume={33},
  number={3},
  pages={965--973},
  year={2022},
  publisher={Springer}
}

@INPROCEEDINGS{Ahmed,
  author={Ahmed Adam, Ahmed Hamed and Chen, Jiawei and Kamel, Salah and Domínguez-García, José Luis},
  booktitle={2024 International Conference on Artificial Intelligence, Computer, Data Sciences and Applications (ACDSA)}, 
  title={Power Management Control for Hybrid Electric Aircraft Propulsion Drive Based on Triple Active Bridge DC-DC Converter}, 
  year={2024},
  volume={},
  number={},
  pages={1-6},
  doi={10.1109/ACDSA59508.2024.10467866}}

@INPROCEEDINGS{syed,
  author={Syed, Wasif H. and Machado, Juan E. and Schiffer, Johannes},
  booktitle={2024 American Control Conference (ACC)}, 
  title={Distributed Adaptive Control for a DC Power Distribution System of a Series-Hybrid-Electric Propulsion System of a Commuter Aircraft}, 
  year={2024},
  volume={},
  number={},
  pages={2598-2603},
  doi={10.23919/ACC60939.2024.10644618}}

@misc{FuelCost,
  author  = {Kelly, S. and Kumar, D. K.},
  title   = {Jet fuel costs could rise from new rules to improve air quality},
  year    = {2019},
  note    = {Reuters}
}

@article{WorldData,
  title={Cars, Planes, Trains: Where do {CO2} Emissions From Transport Come From?},
  author={Hannah Ritchie},
  journal={Our World in Data.},
  year={2022}
}

@article{Schefer,
  title={Discussion on electric power supply systems for all electric aircraft},
  author={Schefer, Hendrik and Fauth, Leon and Kopp, Tobias H and Mallwitz, Regine and Friebe, Jens and Kurrat, Michael},
  journal={IEEE Access},
  volume={8},
  pages={84188--84216},
  year={2020},
  publisher={IEEE}
}

@article{XieYe,
  title={Review of hybrid electric powered aircraft, its conceptual design and energy management methodologies},
  author={Ye, XIE and Savvarisal, Al and Tsourdos, Antonios and Zhang, Dan and Jason, GU},
  journal={Chinese Journal of Aeronautics},
  volume={34},
  number={4},
  pages={432--450},
  year={2021},
  publisher={Elsevier}
}

@inproceedings{Braitor,
  title={Control of DC power distribution system of a hybrid electric aircraft with inherent overcurrent protection},
  author={Braitor, AC and Mills, AR and Kadirkamanathan, V and Konstantopoulos, GC and Norman, PJ and Jones, CE},
  booktitle={2018 IEEE International Conference on Electrical Systems for Aircraft, Railway, Ship Propulsion and Road Vehicles \& International Transportation Electrification Conference (ESARS-ITEC)},
  pages={1--6},
  year={2018},
  organization={IEEE}
}

@inproceedings{ZhangSau,
  title={Comparison of technical features between a more electric aircraft and a hybrid electric vehicle},
  author={Zhang, He and Saudemont, Christophe and Robyns, Beno{\^\i}t and Petit, Marc},
  booktitle={2008 IEEE Vehicle Power and Propulsion Conference},
  pages={1--6},
  year={2008},
  organization={IEEE}
}

@article{Magne,
  title={Active stabilization of {DC} microgrids without remote sensors for more electric aircraft},
  author={Magne, Pierre and Nahid-Mobarakeh, Babak and Pierfederici, Serge},
  journal={IEEE Transactions on Industry Applications},
  volume={49},
  number={5},
  pages={2352--2360},
  year={2013},
  publisher={IEEE}
}

@article{MengL2,
  title={Review on control of {DC} microgrids and multiple microgrid clusters},
  author={Meng, Lexuan and Shafiee, Qobad and Trecate, Giancarlo Ferrari and Karimi, Houshang and Fulwani, Deepak and Lu, Xiaonan and Guerrero, Josep M},
  journal={IEEE Journal of Emerging and Selected Topics in Power Electronics},
  volume={5},
  number={3},
  pages={928--948},
  year={2017},
  publisher={IEEE}
}

@inproceedings{Machado1,
  title={A passivity-inspired design of power-voltage droop controllers for {DC} microgrids with electrical network dynamics},
  author={Machado, Juan E and Schiffer, Johannes},
  booktitle={2020 59th IEEE CDC},
  pages={3060--3065},
  year={2020},
  organization={IEEE}
}

@article{Trip2,
  title={Distributed averaging control for voltage regulation and current sharing in {DC} microgrids},
  author={Trip, Sebastian and Cucuzzella, Michele and Cheng, Xiaodong and Scherpen, Jacquelien},
  journal={IEEE Control Systems Letters},
  volume={3},
  number={1},
  pages={174--179},
  year={2018},
  publisher={IEEE}
}

@inproceedings{Nahata,
  title={On existence of equilibria, voltage balancing, and current sharing in consensus-based {DC} microgrids},
  author={Nahata, Pulkit and Ferrari-Trecate, Giancarlo},
  booktitle={2020 ECC},
  pages={1216--1223},
  year={2020},
  organization={IEEE}
}

@article{Machado2,
  title={Online Parameters Estimation Schemes to Enhance Control Performance in {DC} Microgrids},
  author={Machado, Juan E and Rinaldi, Gianmario and Cucuzzella, Michele and Menon, Prathyush P and Scherpen, Jacquelien MA and Ferrara, Antonella},
  journal={European Journal of Control},
  pages={100860},
  year={2023},
  publisher={Elsevier}
}

@article{Noroozi2,
  title={Model predictive control of {DC} microgrids: current sharing and voltage regulation},
  author={Noroozi, Navid and Trip, Sebastian and Geiselhart, Roman},
  journal={IFAC-PapersOnLine},
  volume={51},
  number={23},
  pages={124--129},
  year={2018},
  publisher={Elsevier}
}

@article{arjan_book_port_Ham,
  title={Port-Hamiltonian systems theory: An introductory overview},
  author={Van Der Schaft, Arjan and Jeltsema, Dimitri and others},
  journal={Foundations and Trends{\textregistered} in Systems and Control},
  volume={1},
  number={2-3},
  pages={173--378},
  year={2014},
  publisher={Now Publishers, Inc.}
}

@article{astolfi_ortega,
  title={Dynamic extension is unnecessary for stabilization via interconnection and damping assignment passivity-based control},
  author={Astolfi, Alessandro and Ortega, Romeo},
  journal={Systems \& Control Letters},
  volume={58},
  number={2},
  pages={133--135},
  year={2009},
  publisher={Elsevier}
}

@article{canseco_ortega_IDA_PBC,
  title={Interconnection and damping assignment passivity-based control: A survey},
  author={Ortega, Romeo and Garcia-Canseco, Eloisa},
  journal={European Journal of Control},
  volume={10},
  number={5},
  pages={432--450},
  year={2004},
  publisher={Elsevier}
}

@book{arjan_book_l2S_3ed,
  title={L2-gain and passivity techniques in nonlinear control},
  author={Van der Schaft, Arjan},
  year={2017},
  publisher={Springer}
}

@article{HandongBai,
title = {Voltage regulation and current sharing for multi-bus DC microgrids: A compromised design approach},
journal = {Automatica},
volume = {142},
pages = {110340},
year = {2022},
issn = {0005-1098},
author = {Handong Bai and Hongwei Zhang and He Cai and Johannes Schiffer}
}

@article{krishna_power-hardware---loop_2022,
	title = {A power-hardware-in-the-loop testbed for intelligent operation and control of low-inertia power systems},
	volume = {70},
	issn = {2196-677X},
	language = {en},
	number = {12},
	urldate = {2023-05-31},
	journal = {at - Automatisierungstechnik},
	author = {Krishna, Ajay and Jaramillo-Cajica, Ismael and Auer, Sabine and Schiffer, Johannes},
	month = dec,
	year = {2022},
	note = {Publisher: De Gruyter (O)},
	pages = {1084--1095},
}

\newpage

\vfill

\end{document}